\documentclass[journal,comsoc]{IEEEtran}

\usepackage{color}
\usepackage[table]{xcolor}

\usepackage{amsmath}
\usepackage{multirow}
\usepackage{amssymb}

\usepackage{algorithm}
\usepackage[noend]{algpseudocode}

 \usepackage{amsthm}

\newtheorem{proposition}{Proposition}

\usepackage[font=small,labelfont=bf]{caption}

\usepackage{subcaption}
 \usepackage{graphicx}

\usepackage{cite}
\ifCLASSINFOpdf
\else
\fi

\begin{document}
%
\title{UxNB-Enabled Cell-Free Massive MIMO with HAPS-Assisted Sub-THz Backhauling}
\author{Omid~Abbasi, \IEEEmembership{Senior Member,~IEEE}, Halim~Yanikomeroglu, \IEEEmembership{Fellow,~IEEE} 
\thanks{O. Abbasi and H. Yanikomeroglu are with
the Non-Terrestrial Networks (NTN) Lab, Department of Systems and Computer Engineering, Carleton University, Ottawa, ON K1S5B6, Canada. e-mail: omidabbasi@sce.carleton.ca; halim@sce.carleton.ca (\textit{Corresponding author: Omid Abbasi})}
\thanks{A preliminary version of this work appeared in the Proceedings of the 2022 IEEE International Conference on Communications (ICC) \cite{conf_version}.}}

\maketitle

\begin{abstract}
In this paper, we propose a cell-free scheme for unmanned aerial vehicle (UAV) base stations (BSs) to manage the severe intercell interference between terrestrial users and UAV-BSs of neighboring cells. Since the cell-free scheme requires enormous bandwidth for backhauling, we propose to use the sub-terahertz (sub-THz) band for the backhaul links between UAV-BSs and central processing unit (CPU). Also, because the sub-THz band requires a reliable line-of-sight link, we propose to use a high altitude platform station (HAPS) as a CPU. At the first time-slot of the proposed scheme, users send their messages to UAVs at the sub-6 GHz band. The UAVs then apply match-filtering and power allocation. At the second time-slot, at each UAV, orthogonal resource blocks are allocated for each user at the sub-THz band, and the signals are sent to the HAPS after analog beamforming. In the HAPS receiver, after analog beamforming, the message of each user is decoded. We formulate an optimization problem that maximizes the minimum signal-to-interference-plus-noise ratio of users by finding the optimum allocated power as well as the optimum locations of UAVs. Simulation results demonstrate the superiority of the proposed scheme compared with aerial cellular and terrestrial cell-free baseline schemes. 

\end{abstract}

\begin{IEEEkeywords}
UxNB, cell-free, HAPS, THz, backhaul.
\end{IEEEkeywords}

%
\IEEEpeerreviewmaketitle

\section{Introduction}
\subsection{Beckground}
The application of unmanned aerial vehicles (UAVs) as base stations (BSs), which are called UxNBs by the third generation
partnership project (3GPP)  \cite{3GPP_uxnb}, has attracted substantial attention for 5G and beyond-5G due to their many advantages, such as high mobility, on-demand deployment, and a high probability of establishing a line-of-sight (LoS) link with users. UxNBs can provide connectivity for users at special events, such as those held in stadiums or theatres, or in areas impacted by natural disasters like floods or earthquakes. These scenarios can be in served, under-served or un-served areas. 
Deploying UxNBs can also be very useful in cases where terrestrial infrastructure may be unable to serve all users due to a temporary spike in demand. At such cases, some users can be offloaded to the aerial infrastructure. 
While LoS connectivity between each UxNB and users can provide very high data rates for the users within cell, severe intercell interference in aerial cellular networks \cite{comp_in_sky}, which is caused by UxNBs of neighboring cells, is a big problem.
 Another challenge of utilizing UxNBs is their backhauling, and since this can not be performed through the fiber link; it must be wireless \cite{Elham_conf}. The backhauling challenge of UxNBs is even greater in remote areas, where terrestrial infrastructure or fiber links may be lacking.  In this paper, we solve these two challenges of UxNBs, i.e., intercell interference and backhauling, with the aid of the cell-free \cite{Elina} scheme and a High altitude platform station (HAPS) \cite{survey_haps}, respectively.

The cell-free massive multiple-input multiple-output (MIMO) technology has emerged as a promising solution for the advancement of wireless communication beyond 5G \cite{Elina}. In traditional cellular networks with co-located massive MIMO, each user is connected to a single BS equipped with a massive number of antennas. Nonetheless, in the cell-free massive MIMO, each user is connected to a large number of access points (APs). In this case, all APs are connected to a central processing unit (CPU) \cite{ngo}, where the received signals from all APs are combined. Indeed, unlike the conventional cellular scheme where received signals from neighboring cells are typically treated as interference, the cell-free scheme leverages these received signals as valuable information at the CPU. This allows for the detection of transmitted signals from users by leveraging the received signals from multiple APs and combining them at the CPU.

\subsection{State of the Art}
The application of HAPSs in wireless networks has attracted a lot of attention recently \cite{survey_haps,3GPP_haps,Softbank}. HAPSs are typically deployed in the stratosphere at an altitude of around $20 ~\mathrm{km}$ with a quasi-stationary position relative to the earth \cite{grace2011broadband}. A HAPS can provide LoS communication and a wide
coverage radius of $50-500 ~\mathrm{km}$, and it can be equipped with powerful computing resources and batteries \cite{survey_haps}. In \cite{Sahabul_haps}, the authors envisioned a HAPS as a super macro base station to provide connectivity in a plethora of applications. Unlike a conventional HAPS, which targets broad coverage for remote areas or disaster recovery, they envisioned HAPS for highly populated metropolitan areas. In \cite{ren2021caching}, HAPS computing was considered as a promising extension of the edge computing.  
In \cite{Safwan}, the authors analyzed the link budget of the aerial platforms equipped with reconfigurable smart surfaces, and compared their communication performance with that of the terrestrial networks.

In the cell-free scheme, enormous bandwidth is required for the backhaul links between APs and CPU. For terrestrial cell-free APs, this huge bandwidth for backhauling can be provided by fiber links \cite{ngo,Manijeh}. However, for our proposed aerial cell-free APs, this backhauling must be wireless. In order to satisfy the enormous bandwidth requirements of backhauling for UxNBs, we need to utilize the upper frequency bands for these wireless links \cite{Elham_conf}. In \cite{Elham_caching}, 
in order to address the limited wireless backhaul
capacity of UxNBs and consequently, decrease the latency,
content caching is proposed to alleviate the backhaul
congestion. The authors in \cite{backhaul_number} provided analytical expressions for the probability of
successfully establishing a backhaul link in the millimeter-wave band between UAV-BSs and ground stations, and they showed that increasing the density of
the ground station network improved the
performance of the backhauling. In \cite{comp_in_sky}, the authors proposed utilizing the coordinate multipoint scheme for UAV-BSs in uplink communications. They assumed that the backhaul links between all UAVs and the CPU were perfect so that the
signal distortion induced by the backhaul transmission was ignored.
The problem of wireless backhauling of UxNBs is largely due to the dynamic blockages and shadowing between UxNBs and a terrestrial CPU, which makes it difficult to utilize the upper frequency bands (such as the terahertz (THz) band) for these links \cite{mmwave_UAV_backhaul}. Higher frequency bands require a reliable LoS link, and probabilistic LoS links between UAVs and a terrestrial CPU is not suitable for these bands. We propose to utilize a HAPS in the stratosphere to solve this problem. 

The THz band is generally defined as the region of the electromagnetic
spectrum in the range of $100 ~\mathrm{GHz}$ to $10 ~\mathrm{THz}$, and the sub-THz band is defined as the frequencies in the range of $100 ~\mathrm{GHz}$ to $300 ~\mathrm{GHz}$ \cite{THz_loss_Mag,akyildiz2014terahertz}.
The D band ($110-170 ~\mathrm{GHz}$) is among the next interesting range of
frequencies for beyond-5G systems \cite{D-band-juntti,Rappa}, and hence we consider this band as the carrier frequency in our paper. The authors in \cite{Petrov_SINR_THz}
 developed an analytical model for interference
and signal-to-interference-plus-noise ratio (SINR) assessment in dense THz networks obtaining the
first two moments and density functions for both metrics. 
In \cite{dahrouj}, the authors investigated
a THz Ultra-Massive-MIMO-based aeronautical communication scheme for the space-air-ground integrated network. In \cite{UAV_THz}, the problem of UAV deployment, power allocation, and bandwidth allocation was investigated for a UAV-assisted wireless system operating at THz frequencies. 

\subsection{Motivations and Contributions}
Our work proposes the application of the cell-free scheme to a set of aerial access points (i.e., UxNBs) to address the challenges posed by severe intercell interference in aerial cellular networks, particularly between UxNBs of neighboring cells and terrestrial users. In contrast to terrestrial networks, where the transmitted signals face substantial non(N)-LoS path loss to reach distant APs, aerial networks exhibit a distinct behavior. With the presence of LoS links between terrestrial users and aerial APs, the transmitted signals in aerial networks can reach even far APs. The cell-free scheme effectively transforms these strong destructive interference signals into valuable signals that can be utilized by the CPU for the detection of users' messages.  While numerous studies have investigated the cell-free scheme in terrestrial networks \cite{Elina,ngo,Manijeh,Bashar_backhaul}, to the best of our knowledge, our research is the first to consider the application of the cell-free scheme specifically for UAV-BSs.

In our scheme, instead of a terrestrial CPU, we propose to utilize a HAPS as an aerial CPU to process all the received signals from all UxNBs. HAPS is an ideal choice to work as a CPU for our proposed cell-free scheme since there is negligible blockage and shadowing for backhaul links between a HAPS and UxNBs, which means the LoS links will be reliable. Hence, we can easily use the upper frequency bands for these links to support the enormous bandwidth requirement for backhauling of the proposed cell-free scheme.  In this paper, we propose to use the sub-THz frequency band for the backhaul links between UxNBs and HAPS.
The D band ($110-170 ~\mathrm{GHz}$) is among the next interesting range of
frequencies for beyond-5G systems \cite{D-band-juntti,Rappa}, and hence we consider this band as the carrier frequency in our paper. \footnote{In addition to backhauling of the UxNBs in urban and dense urban environments, another important scenario for using a HAPS as a CPU for backhauling of the aerial APs is for the cases where these APs are deployed to serve users in remote areas or where terrestrial infrastructure and fiber links may be lacking or damaged.}

In summary, this paper introduces an uplink cell-free communication scheme developed to address the issue of intercell interference between terrestrial users and UAV BSs located in adjacent cells. To cope with the substantial bandwidth demands associated with the wireless backhaul of UAV BSs in our proposed cell-free scheme, we propose harnessing the sub-THz band to establish high-capacity links between the UAV-BSs and the CPU. Additionally, to ensure reliable LoS backhaul links within the sub-THz band, we introduce a HAPS as the aerial CPU responsible for aggregating the transmitted signals from UAVs.
Our proposed scheme is performed in two time slots. During the initial slot, users transmit their messages to the UAVs using the sub-6 GHz band. The UAVs utilize match-filtering and power allocation techniques to process the received signals. In the subsequent slot, each UAV applies analog beamforming and transmits its signal in the sub-THz band to the HAPS. At the HAPS receiver, analog beamforming is employed for user message decoding.
For the purpose of optimizing system performance, we formulate an optimization problem focused on maximizing the minimum SINR for the users. This problem entails determining the optimal power allocation and UAV locations.
The cell-free scheme presented in this paper, which leverages the sub-THz band and HAPS technology, offers an effective solution to mitigate interference and enhance communication performance in UAV systems.

 The main contributions of this paper are summarized as follows:
\begin{itemize}
 \item A cell-free scheme for a set of aerial APs (UxNBs) is proposed to manage the severe intercell interference in aerial cellular networks between UxNBs of neighboring cells and terrestrial users. To the best of our knowledge, our work is the first to consider the cell-free scheme for the UAV-BSs. 
  \item We utilize a HAPS as a CPU for backhauling of UxNBs in the sub-THz band. In this paper, instead of a terrestrial CPU, we show how a HAPS can be used as an aerial CPU to process all received signals from all UxNBs. HAPS is an ideal choice to work as a CPU since there is negligible blockage and shadowing for backhaul links between it and the UxNBs which means a reliable LoS link. Hence, we can easily use the upper frequency bands for these links to support the huge bandwidth requirement for backhauling.
  \item A transceiver scheme at the UxNBs is proposed. At the first time slot of the proposed cell-free scheme, users send their messages to UxNBs at the sub-6 GHz frequency band. Then each UxNB applies match-filtering to align the received signals from users, followed by power allocation among the aligned signals of all users. At the second time slot, at each UxNB, we allocate orthogonal resource blocks (RBs) for each user at the sub-THz band, and forward the filtered signals of all users to the HAPS after analog beamforming.
  
  \item A receiver scheme at the HAPS is proposed. At the HAPS, in order to align the received signals for each user from different UxNBs, we perform analog beamforming. Then, we demodulate and decode the message of each user at its own unique RB. 

  \item We derive a closed-form expression for the achievable rate of the users utilizing the use-and-then-forget bound \cite{marzetta2016fundamentals} based on the proposed transceiver and receiver schemes.
  \item We formulate an optimization problem that maximizes the minimum SINR of users. We find optimum values for two blocks of optimization variables (i.e., the allocated powers for users in each UxNB and the locations of UxNBs), which are solved by the bisection \cite{boyd2020disciplined} and successive convex approximation (SCA) \cite{boyd2004convex} methods, respectively. Finally, the whole optimization problem is solved by the block coordinate descent (BCD) method \cite{razaviyayn2013unified}.

 \end{itemize}
 
Simulation results demonstrate the superiority of the proposed cell-free scheme compared with the aerial cellular and terrestrial cell-free baseline schemes in urban, suburban, and dense urban environments. Also, simulation results show that utilizing a HAPS as a CPU is useful when the considerable path loss in the sub-THz band between UxNBs and HAPS is compensated for by a high number of antenna elements at the HAPS.

\subsection{Outline}
The remainder of this paper is organized as follows. Section II presents the system model. Section III presents the proposed transceiver scheme and the corresponding achievable rate. Section IV provides the formulated optimization problem and its solution for the powers and locations of UxNBs. Section V provides simulation
results to validate the performance of the proposed scheme. Finally, Section VI concludes the paper.

\section{System Model and Channel Model}
In this paper, our aim is to provide terrestrial users with connectivity, utilizing aerial BSs operating in cell-free mode. Communication is conducted in the uplink mode, and we employ a HAPS as a CPU for backhauling the aerial BSs in the sub-THz band. In this section, we provide a detailed explanation of the system model and the channel model of the proposed system. In the subsequent sections, we first present transceiver schemes for both the aerial BSs and HAPS, followed by deriving a closed-form expression for the achievable user rates. We then formulate an optimization problem that maximizes the minimum SINR of users by determining optimal values for two sets of optimization variables, namely the allocated powers for users in each aerial BS and the locations of UAVs. \par
\subsection{System Model}
\begin{figure}[!t]
  \centering
  \includegraphics[width=0.5\textwidth]{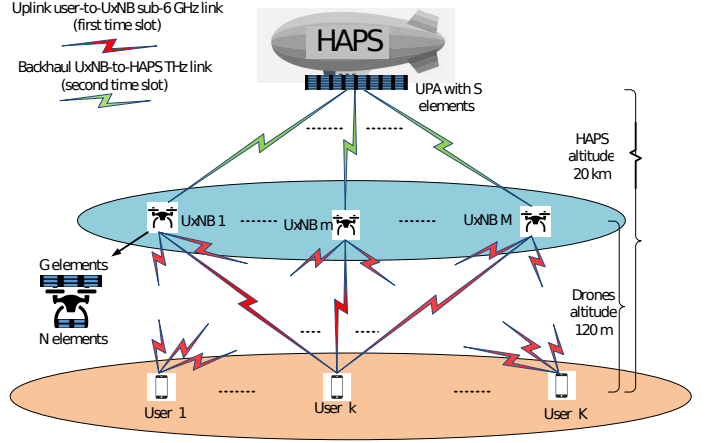}
  \caption{System model for the proposed uplink cell-free scheme with HAPS-assisted sub-THz backhauling. At the first time slot, users send their messages to UxNBs, and each UxNB applies match filtering for its received signals. At the second time slot, the UxNBs forward each user's filtered signal to the HAPS at orthogonal RBs. Then, HAPS decodes the message of each user as a CPU. }\label{system-model}
  \end{figure}
 The proposed aerial cell-free scheme with HAPS-assisted sub-THz backhauling is shown in Fig. \ref{system-model}.
As we can see, $M$ UxNBs are serving $K$ users in the cell-free mode. Each UxNB is assumed to be equipped with a uniform planar array (UPA) with $N$ receive antenna elements positioned on the underside of each UAV working in the sub-6 GHz frequency band, and a UPA with $G$ transmit antenna elements on the topside of the UAV working in the sub-THz band. We propose to utilize a HAPS as a CPU to combine the received signals of all UxNBs and decode the messages of all $K$ users. The HAPS is equipped with a UPA with $S$ receive antenna elements working in the sub-THz band.
Due to the requirement of $K$ orthogonal RBs for retransmission of the received signals at the UxNBs to the HAPS, we propose to use the sub-THz frequency band for the backhaul links.\par
Note that THz communication heavily relies on LoS propagation due to the high-frequency characteristics. However, when the LoS path is obstructed, the achievable rate of the THz channel is significantly limited, posing challenges in serving users without a clear LoS path. In our scenario, we specifically employ the THz band for the UxNB to HAPS links. Given the high flight altitudes of both UAVs and HAPS, the probability of encountering blocking objects between them is extremely low. This enables us to leverage the THz band for efficient and reliable communication between UxNBs and HAPS.\par
In the proposed scheme, the transmission is performed in two time slots. At the first time slot, users send their data to the UxNBs. 
It should be noted that the severe intercell interference among UAVs of neighboring cells and users is a big problem in aerial networks. In order to solve this problem, we propose an aerial cell-free scheme where each user is served by multiple UxNBs that establish a strong link with users. 
At each UxNB, the channel state information (CSI) of the  user-to-UxNB links are estimated, and then they are utilized for match-filtering of the received signals. We also divide the total power of each UxNB among the users. 
Since the instantaneous CSI only needs to be known locally at the UxNBs for the match-filtering scheme, this is a big advantage compared to other schemes, such as zero-forcing, which requires the instantaneous CSI of all links at the CPU \cite{emil}.\par
At the second time slot, the UxNBs forward the power-allocated and match-filtered signals to the HAPS. Here, we utilize the sub-THz band for the UxNB to HAPS link, and we allocate orthogonal resource RBs for each user's signal at these backhaul links. This means that in the cell-free scheme, we need $K$ times more bandwidth for backhauling compared to the access network, which can be satisfied in the sub-THz band. 
 In order to align the received signals for each user from the different UxNBs, we perform analog beamforming based on the steering vectors of the UPA at the HAPS for all UxNBs. Finally, we demodulate and decode the message of each user at its own unique RBs. 
  A summary of the notations used in this article is summarized in Table \ref{not_def}. 

  The beam squint effect, which involves the physical deviation of the beam at various sub-carrier frequencies, results in a notable reduction in array gain within the beamforming architecture. The utilization of ultra-wide bandwidth and highly focused pencil beams in THz systems exacerbates the beam squint effect \cite{THz_6g_Precoding}. Consequently, beams operating at different frequencies may diverge and direct themselves towards distinct directions that are considerably distant from the intended target user. In our proposed scheme, we do not employ a wide bandwidth in the sub-THz band.  Notably, both the first (sub-6 GHz) and second (sub-THz) hops utilize the same bandwidth of 1 MHz. Consequently, our scheme does not experience the beam squint effect, as the limited bandwidth minimizes significant variations in the physical direction of the beam across different sub-carrier frequencies.

\begin{table*}[t]
\small
	\caption{The definition of notations.}
	\vspace{-11pt}
	\begin{center}
	{\rowcolors{2}{white}{blue!10}
		\begin{tabular}{|c|c|} \hline \label{not_def}
		\textbf{Notation} & 	\textbf{Definition} \\ 
			\hline \hline
		$M$ and $K$  &  Number of UxNBs and users \\ 
			$N$ and $G$ & Total number of receive and transmit antenna elements at each UxNB\\ 
	$S$ & Total number of receive antenna elements at HAPS\\ 
			$h_{kmn}$ & The channel between user $k$ and antenna element $n$ of UxNB $m$\\ 
			$\mathbf{a}_{km}=[a_{kmn}]_{1\times N}$ & The steering vector of the receive antenna array of UxNB $m$ for user $k$\\ 
   $\bold{b}_{m}=[b_{mg}]_{1\times G}$ & The steering vector of the transmit antenna array of UxNB $m$\\
    $\bold{c}_{m}=[c_{ms}]_{1\times S}$ & The steering vector of the receive antenna array of the HAPS transmitted from UxNB $m$\\
			$d_{km}$ & Distance between user $k$ and UxNB $m$ \\ 
				$\lambda$ and $f$ &  Wavelength and carrier frequency  \\ 
			$(x_{u,k}, y_{u,k},0)$ & The coordinates of each user $k$  \\ 
   $(x_{d,m}, y_{d,m}, z_{d,m})$ & The coordinates of each UxNB $m$\\
			$\theta_{km}$ and $\phi_{km}$  & The elevation and azimuth angles of arrival of the transmitted signal from user $k$ at the UxNB $m$\\
   $\Theta_{m}$ and $\Phi_{m}$  & The elevation and azimuth angles of the transmitted signal from UxNB $m$ at the HAPS\\
			$\mathsf{PL}_{km}^{\mathsf{LoS}}$ and $\mathsf{PL}_{km}^{\mathsf{NLoS}}$ & LoS and NLoS path loss of the link between UxNB $m$ and user $k$ \\
			$\mathsf{FSPL}_{km}$ & The free-space path loss between user $k$ and UxNB $m$ \\ 
   $\eta_{\mathsf{LoS}}^{\mathsf{dB}}$ and $\eta_{\mathsf{NLoS}}^{\mathsf{dB}}$ & The excessive path losses  affecting the air-to-ground links for LoS and NLoS cases\\
   $A$ and $B$ & Fixed parameters that determine the probability of existing LoS link \\
   $\beta_{km}^2$ & The large-scale channel power gain between user $k$ and UxNB $m$\\
   $g_{mgs}$ & The channel between the element $g$ of UxNB $m$ and the element $s$ of the HAPS\\ 
   $d_{m}$ & The distance between the reference element of UxNB $m$ and the reference element of HAPS\\
$\gamma_m^2$ and $\rho_m^2$& The total and free space path loss between UxNB $m$ and the HAPS\\
  $\tau_{m}$ and $K_a$ ($\mathrm{dB/km}$) & The transmittance and absorption coefficient of the medium  \\
 $h_m^{e}$ and $h_{e}$ & The effective height of a medium for UxNB $m$ and for a UxNB in the nadir of the HAPS\\
 $\omega$ & A random variable with uniform distribution as $U(0,2\pi)$\\
 $P_k$ and  $s_k$ & The maximum transmit power and the transmitted symbol at user $k$\\
  $P_m$ & The total power at UxNB $m$\\
   $Z_{m}$ and $Z_H$ & The AWGN noise at the receiver of UxNB $m$ and HAPS\\
        $\eta$, $\bold{T}$,  $\bold{t}$,  and $\zeta$ & Slack variables to solve optimization problems\\
  $\epsilon$ &  The tolerance value in Algorithm 1\\
			\hline
		\end{tabular}}
	\end{center}
	\vspace{-19pt}
 \normalsize
\end{table*}

\subsection{Channel Model}
In our scheme, the channel between user $k$ and antenna element $n$ of UxNB $m$ is indicated by $h_{kmn}$, which includes both large-scale fading (i.e., path loss and shadowing) and multipath small-scale fading effects. We assume a UPA for the receiver of each UAV with $N=N_w\times N_l$ antenna elements, where $N_w$ and $N_l$ show the number of antenna elements in the width and length of the array, respectively. Because of the possibility of the existence of a LoS link between the users and UAVs, a Ricean distribution is considered for the channel between user $k$ and antenna element $n=(n_w,n_l)$ of UxNB $m$ as follows: 
\small
\begin{equation}\label{prob_hkmn}
    h_{kmn}=10^{-\dfrac{\mathsf{PL_{km}}}{20}}(\sqrt{P_{km}^{\mathsf{LoS}}}a_{kmn}+\sqrt{P_{km}^{\mathsf{NLoS}}}CN(0,1)),
\end{equation}
\normalsize
where $CN(0,1)$ shows a complex normal random variable with a mean value of 0 and a variance (power) of 1. Also, 
\small
\begin{equation}
\begin{split}
      a_{kmn}=\exp(j2\pi(\frac{d_{km}}{\lambda_\mathsf{sub6}}))&\times\exp(j2\pi(\frac{d_\mathsf{sub6,w}(n_w-1)\sin\theta_{km}\cos\phi_{km}}{\lambda_\mathsf{sub6}}))\\&\times\exp(j2\pi(\frac{d_\mathsf{sub6,l}(n_l-1)\sin\theta_{km}\sin\phi_{km}}{\lambda_\mathsf{sub6}}))
\end{split}
\end{equation} 
\normalsize
indicates the phase shift of the LoS link's signal due to distance in which $d_{km}$ shows the distance between user $k$ and UxNB $m$; $d_\mathsf{sub6,w}=\frac{\lambda_\mathsf{sub6}}{2}$ ($d_\mathsf{sub6,l}=\frac{\lambda_\mathsf{sub6}}{2}$) is the element spacing along the width (length) of antenna array for each UxNB in sub-6 GHz frequency band $f_\mathsf{sub6}$; $\lambda_\mathsf{sub6}=\frac{C}{f_\mathsf{sub6}}$ is the wavelength, and $C=3\times10^8 ~\mathrm{m/s}$ is the speed of light. 
The coordinates of each user $k$ and each UAV $m$ are denoted by $(x_{u,k}, y_{u,k},0)$ and $(x_{d,m}, y_{d,m}, z_{d,m})$, respectively. Hence the distance between user $k$ and UAV $m$ equals $d_{km}=\sqrt{(x_{u,k}-x_{d,m})^2+(y_{u,k}-y_{d,m})^2+z_{d,m}^2}$.
$\theta_{km}$ and $\phi_{km}$  show the elevation and azimuth angles of arrival of the transmitted signal from user $k$ at the UxNB $m$, respectively.
It is worth mentioning that $\bold{a}_{km}=[a_{kmn}]_{1\times N}$ creates the steering vector of the receive antenna array of UAV $m$ for user $k$.  
Also, we have $\mathsf{PL_{km}}=P_{km}^{\mathsf{LoS}}\mathsf{PL_{km}^{\mathsf{LoS}}}+P_{km}^{\mathsf{NLoS}}\mathsf{PL_{km}^{\mathsf{NLoS}}}$ in which the LoS and non-LoS (NLoS) path loss of the link between UxNB $m$ and user $k$ are equal to $\mathsf{PL}_{km}^{\mathsf{LoS}}=\mathsf{FSPL}_{km}+\eta_{\mathsf{LoS}}^{\mathsf{dB}}$, and $\mathsf{PL}_{km}^{\mathsf{NLoS}}=\mathsf{FSPL}_{km}+\eta_{\mathsf{NLoS}}^{\mathsf{dB}}$, respectively \cite{Hourani}. In these equations, $\mathsf{FSPL}_{km}=10\log(\frac{4\pi f_\mathsf{sub6}d_{km}}{C})^2$ shows the free-space path loss (FSPL), and $\eta_{\mathsf{LoS}}^{\mathsf{dB}}$ and $\eta_{\mathsf{NLoS}}^{\mathsf{dB}}$ indicate the excessive path losses (in dB)  affecting the air-to-ground links for LoS and NLoS cases, respectively \cite{Irem}. $P_{km}^{\mathsf{LoS}}=\frac{1}{1+A\exp({-B(90-\theta_{km})-A})}$ shows the probability of establishing LoS link between user $k$ and UxNB $m$  in which $\theta_{km}$ (in degree) shows the elevation angle between user $k$ and UxNB $m$, and $A$ and $B$ are parameters depending on the environment \cite{Hourani}. $P_{km}^{\mathsf{NLoS}}=1- P_{km}^{\mathsf{LoS}}$ shows the probability of establishing a NLoS link between user $k$ and UxNB $m$.

The large-scale channel power gain for the user $k$ to UxNB $m$ link is equal to
\small
\begin{equation}\label{bkm_formula}
\begin{split}
    \beta_{km}^2=E\{|h_{kmn}|^2\}&=E\{h_{kmn}h_{kmn}^*\}=10^{-\frac{\mathsf{PL_{km}}}{10}}\\&=10^{-\frac{P_{km}^{\mathsf{LoS}}\mathsf{PL_{km}^{\mathsf{LoS}}}+P_{km}^{\mathsf{NLoS}}\mathsf{PL_{km}^{\mathsf{NLoS}}}}{10}}.
\end{split}
\end{equation}
\normalsize
By considering $\beta_{0}=(\frac{4\pi f_\mathsf{sub6}}{C})^{-2}$ as the channel gain at the reference distance $d_{km}=1~\mathrm{m}$, the large-scale channel power gain can be rewritten as
  $\beta_{km}^2=\eta_{km}\beta_{0}(d_{km})^{-2}$,
in which $\eta_{km}=10^{-\frac{P_{km}^{\mathsf{LoS}}\eta_{\mathsf{LoS}}^{\mathsf{dB}}+P_{km}^{\mathsf{NLoS}}\eta_{\mathsf{NLoS}}^{\mathsf{dB}}}{10}}$ shows the excessive path loss.
We consider independent additive white Gaussian noise (AWGN) with the distribution $CN(0,\sigma^{2})$ at all antenna elements of all UxNBs. We assume that all of the antenna elements in this paper are omnidirectional with an antenna gain of 1.

We assume a UPA for the transmitter of each UxNB with $G=G_w\times G_l$ antenna elements in which $G_w$ and $G_l$ show the number of antenna elements along the width and length of the array, respectively. We also assume a UPA at the receiver of the HAPS with a large number of $S=S_w\times S_l$ antenna elements in which $S_w$ and $S_l$ show the number of antenna elements along the width and length of the array, respectively.

\begin{figure*}[!t]
  \centering
  \includegraphics[width=0.9\textwidth]{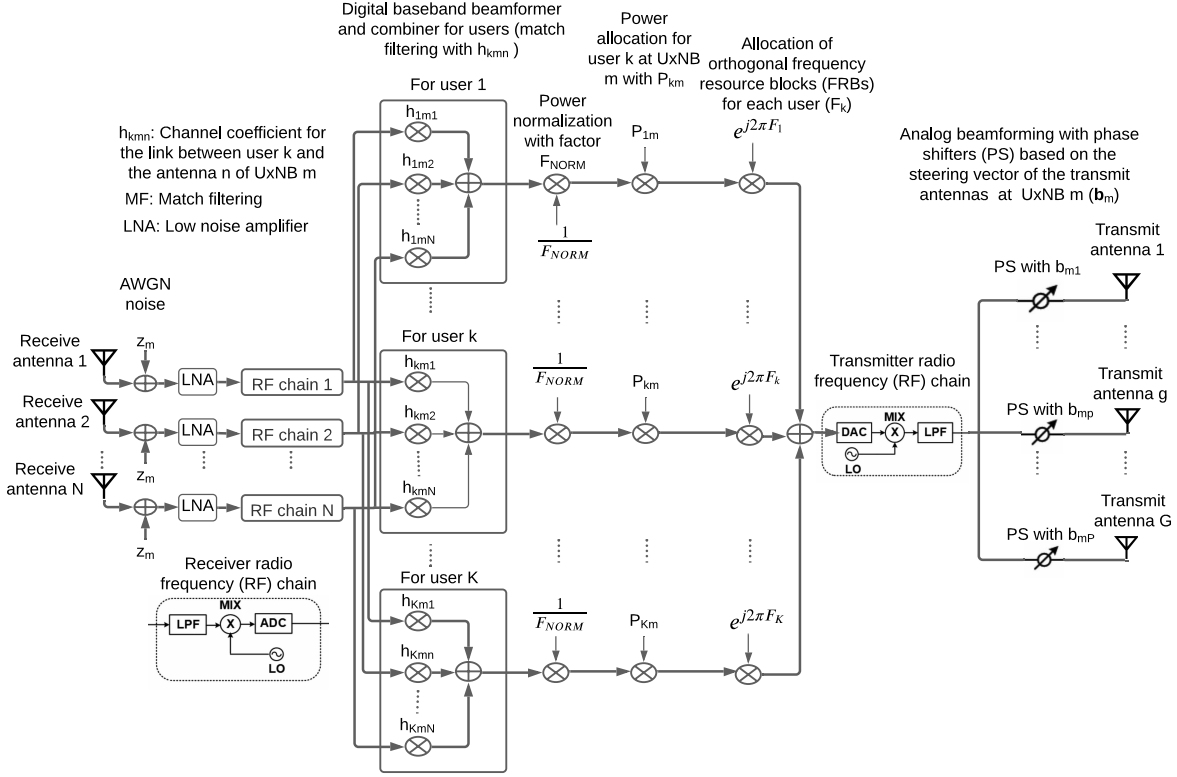}
  \caption{The proposed transceiver scheme at UxNB $m$.}\label{transc_signal}
\end{figure*}

We consider a LoS channel between HAPS and UxNBs based on the following two justifications. First, we have assumed a fixed flight height of 120 meters for all UAVs, which minimizes the presence of scatterers in close proximity to the UAVs. Given that the transmit antennas of the UxNBs are positioned on top of each UAV, the probability of the reflected beams reaching the UxNB antennas is significantly reduced, thereby minimizing the occurrence of NLoS paths. Second, the UxNB to HAPS links operate in the sub-THz band, enabling the creation of highly focused and narrow pencil beams \cite{thz-uav}. Consequently, the probability of encountering NLoS paths diminishes further.
 Therefore, the channel between the transmit antenna element $g=(g_w,g_l)$ of UxNB $m$ and the receiver antenna element $s=(s_w,s_l)$ of the HAPS, is equal to 

\small
\begin{equation}
    g_{mgs}=\gamma_{m}b_{mg}^*c_{ms},
\end{equation}
\normalsize
where
\small
 \begin{equation}\label{bmg}
 \begin{split}
       b_{mg}=\exp(j2\pi(\frac{d_{m}}{\lambda_\mathsf{THz}}))&\times \exp(j2\pi(\frac{d_\mathsf{THz,d,w}(g_w-1)\sin\Theta_{m}\cos\Phi_{m}}{\lambda_\mathsf{THz}}))\\&\times\exp(j2\pi(\frac{d_\mathsf{THz,d,l}(g_l-1)\sin\Theta_{m}\sin\Phi_{m}}{\lambda_\mathsf{THz}}))
 \end{split}
 \end{equation}
\normalsize
  indicates the phase shift of the transmitted signal from antenna element $g$ of UxNB $m$, and where
  \small
  \begin{equation}\label{cms}
  \begin{split}
      c_{ms}=\exp(j2\pi&(\frac{d_\mathsf{THz,h,w}(s_w-1)\sin\Theta_{m}\cos\Phi_{m}}{\lambda_\mathsf{THz}}))\\&\times\exp(j2\pi(\frac{d_\mathsf{THz,h,l}(s_l-1)\sin\Theta_{m}\sin\Phi_{m}}{\lambda_\mathsf{THz}})) 
        \end{split}
  \end{equation}
\normalsize
  indicates the phase shift of the received signal from user $m$ at antenna element $s$ of the HAPS. In these equations, $d_{m}$ indicates the distance between the reference antenna element of UxNB $m$ and the reference antenna element of HAPS; $d_\mathsf{THz,d,w}=\frac{\lambda_\mathsf{THz}}{2}$ ($d_\mathsf{THz,d,l}=\frac{\lambda_\mathsf{THz}}{2}$) is the element spacing along the width (length) of the transmit antenna array at each UxNB in sub-THz frequency band $f_\mathsf{THz}$; $d_\mathsf{THz,h,w}=\frac{\lambda_\mathsf{THz}}{2}$ ($d_\mathsf{THz,h,l}=\frac{\lambda_\mathsf{THz}}{2}$) is the element spacing along the width (length) of the receiver antenna array at the HAPS in the sub-THz frequency band $f_\mathsf{THz}$; and $\lambda_\mathsf{THz}=\frac{C}{f_\mathsf{THz}}$ is the wavelength. Also, $\Theta_{m}$ and $\Phi_{m}$  show the elevation and azimuth angles of the transmitted signal from UxNB $m$ at the HAPS.
It is worth mentioning that $\bold{b}_{m}=[b_{mg}]_{1\times G}$ creates the steering vector of the transmit antenna array of UxNB $m$, and $\bold{c}_{m}=[c_{ms}]_{1\times S}$ creates the steering vector of the receive antenna array of the HAPS transmitted from UxNB $m$. 
 $\gamma_m^2$ shows the path loss between UxNB $m$ and the HAPS. The path loss for the sub-THz band is given by $\gamma_{m}^2=
 \rho_m^2 \tau_m=\gamma_0 d_{m}^{-2} \tau_{m}$ in which $\rho_m^2$ is the free space path loss, $\gamma_0=(\frac{4\pi f_\mathsf{THz}}{C})^{-2}$ 
 is the channel gain at the reference distance $d_{m}=1~\mathrm{m}$,
 and $\tau_{m}=10^{-K_ah_m^{e}/10}$ shows the transmittance of the medium following
the Beer-Lambert law in which $K_a$ (in $\mathrm{dB/km}$) is the absorption coefficient of the medium, which is a function of frequency and altitude \cite{Petrov_SINR_THz}. Also, $h_m^{e}=\frac{h^{e}}{\sin \Theta_m}=\frac{h^{e}d_m}{h_{\mathrm{HAPS}}-z_{d,m}}$ shows the effective height of a medium for UxNB $m$ in which $h_{e}$ indicates the effective height for a UxNB that is located in the nadir of the HAPS \cite{ITU_series2019attenuation}.

It should be noted that the absorbed part of the sub-THz signal by the medium due to the molecule absorption (i.e., $1-\tau_{m}$ percent of the transmitted signal from UxNB $m$) will be re-emitted by the molecules with some delay, and hence we consider this re-emitted signal as re-emission interference in our rate derivations \cite{Saad_thz}.
We assume that this delayed re-emitted signal has a random phase as $\exp(j\omega)$ in which $\omega$ is a random variable with uniform distribution such that $U(0,2\pi)$.  
 We consider independent AWGN with the distribution $CN(0,\sigma_{H}^{2})$ at all antenna elements of the HAPS.

\section{Proposed Transceiver Scheme at UxNBs and HAPS}
We can see the proposed transceiver scheme at UxNB $m$ in Fig. \ref{transc_signal}. At the first time slot of the proposed scheme, each user transmits its message to the UxNBs. The transmitted signal by user $k$ is shown by $\sqrt{P_k}s_k$ in which $P_k$ (for $k\in\{1,2,...,K\}$) indicates the maximum transmit power at each user $k$, and $s_k$ ($E\{|s_{k}|^{2}\}=1$) is the transmitted symbol from user $k$. The received signal at the antenna element $n$ of UxNB $m$'s UPA equals $y_{mn}=\sum_{k=1}^{K}h_{kmn}\sqrt{P_k}s_k+Z_{m}$ in which $Z_{m}$ is the AWGN noise at the receiver of UxNB $m$. 
After receiving $y_{mn}$ at the antenna element $n$ of UxNB $m$, the low noise amplifier (LNA) block amplifies this signal, and then the radio frequency (RF) chains downconvert the signal to a baseband one and convert the analog signal to a digital signal.
Next, at the digital baseband beamforming block for each user and according to the estimated CSI for the channel between user $k$ and antenna element $n$ of UxNB $m$ (i.e., $h_{kmn}$), we perform match-filtering such that $y_{kmn}^{\mathsf{MF}}=y_{mn}\times \frac{h_{kmn}^*}{|h_{kmn}|}$, and we combine these match-filtered signals to arrive at $y_{km}^{\mathsf{COMB}}=\sum_{n=1}^Ny_{kmn}^{\mathsf{MF}}$ \cite{Omid_inspired}. $x^*$ indicates the conjugate of $x$. Then, we allocate the power $P_{km}$ for each user so that we must have $\sum_{k=1}^KP_{km}\leq P_m$ for each UxNB $m$, in which $P_m$ shows the total power at UxNB $m$. We now need to normalize the signal for each user before power allocation such that $y_{km}^{\mathsf{NORM}}=\frac{y_{km}^{\mathsf{COMB}}}{|y_{km}^{\mathsf{COMB}}|}$ in which $|x|$ shows the absolute value of $x$. Therefore, the signal for user $k$ at UxNB $m$ after power normalization and power allocation will be $y_{km}=\sqrt{P_{km}}y_{km}^{\mathsf{NORM}}$. In order to transmit these $K$ signals to the HAPS, we allocate orthogonal frequency RBs for each of them to avoid interference among the filtered signals of different users. 
It should be emphasized that the same frequency RBs are allocated for each user at different UxNBs, which means that we need $K$ RBs at the second time slot of the proposed scheme in total. These RBs can be easily provided at the sub-THz frequency band. The digital signal for users is converted to an analog one and upconverted to the sub-THz frequency band utilizing an RF chain \footnote{In order to transmit signals from each UxNB to the HAPS, we utilize only one RF chain because we apply fully analog beamforming. Also, since we perform digital beamforming (match-filtering) at the receiver of each UxNB with $N$ antenna elements, it needs $N$ RF chains. Hence, each UxNB requires $N+1$ RF chains in total for receiving and transmitting signals.}.
We perform the analog beamforming with phase shifters (PSs) to direct the transmitted signal from each UxNB to the HAPS. This is done by multiplying the signal by $b_{mg}$ in (\ref{bmg}) for each transmit antenna element $g$ of each UxNB $m$. Note that the transmitted signal for user $k$ from antenna element $g$ of UxNB $m$ equals $b_{mg}y_{km}$. \par

\begin{figure*}[!t]
  \centering
 \includegraphics[width=0.75\textwidth]{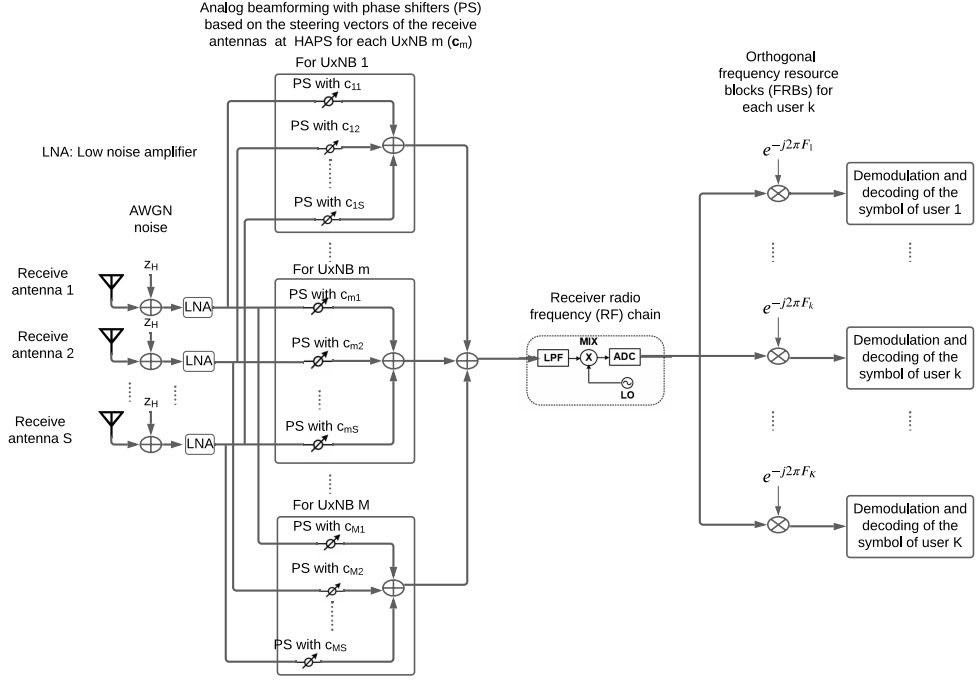}
 \caption{The proposed receiver scheme at HAPS for decoding of the message of users.}\label{receiver_signal}
\end{figure*}

The transmitted signal from each UxNB $m$ will be received at antenna element $s$ of the HAPS after passing through the channel between them (i.e., $g_{mps}$). Hence, the received signal at RB $k$ and antenna element $s$ of the HAPS is given by
\small
\begin{equation}
    \begin{split}
        y_{s}^k=\sum_{m=1}^M \sum_{g=1}^G g_{mgs}b_{mg}y_{km}+Z_H&=\sum_{m=1}^M \sum_{g=1}^G c_{ms}\gamma_{m}b_{mg}b_{mg}^*y_{km}+Z_H\\&=G\sum_{m=1}^M c_{ms}\gamma_{m}y_{km}+Z_H
    \end{split}
\end{equation}
\normalsize
 in which $Z_H$ is the AWGN noise at each receive antenna element of the HAPS. Superscript $k$ shows the received signal at RB $k$. 
We can see the proposed receiver scheme at HAPS for decoding the message of users in Fig. \ref{receiver_signal}. As one can see, first the LNA blocks amplify the received signals. 
In the proposed receiver scheme at the HAPS, we perform analog beamforming with PSs to align the received signals from each UxNB $m$ at the receive antenna elements of the HAPS. To do this, we multiply the signal $y_{s}^k$ by the conjugate of the steering vector of the receive antenna elements at HAPS for each UxNB $m$, i.e., $\bold{c}_{m}^*$  in (\ref{cms}), and then combine these signals as follows:
\small
\begin{equation}\label{yk}
    y^k=\sum_{m=1}^M\sum_{s=1}^S c_{ms}^* y_{s}^k.
\end{equation}
\normalsize
 Next, an RF chain downconverts the signal to a baseband signal, and converts the analog signal to a digital signal. In order to receive and decode signals of all UxNBs at the HAPS, we utilize only one RF chain since we apply fully analog beamforming.
As mentioned above, we allocate orthogonal RBs for each user at the UxNB-to-HAPS sub-THz links, and for this reason, the received signal for each user can be separately achieved at the HAPS. Finally, utilizing this signal $y^k$, we demodulate and decode the symbol of each user $k$.\par

Now, in the following proposition, we derive the $\mathsf{SINR}$ of each user utilizing the signal $y^k$.

\begin{proposition}\label{propos_sinr}
The achievable rate of the user $k$ in the proposed aerial cell-free scheme with HAPS-assisted backhauling in the THz band is given by
\small
\begin{equation}\label{rate-rev}
   R_k=\log_2(1+\mathsf{SINR}_k) 
\end{equation}
\normalsize
in which $\mathsf{SINR}_k$ is given by (\ref{SINR}) on the top of the next page.
\begin{figure*}
\small
\begin{equation}    \label{SINR}
     \mathsf{SINR}_k=\frac{MG^2NSP_{k}(\sum_{m=1}^M \gamma_{m}\sqrt{P_{km}} \beta_{km})^2}{MG^2\sum_{m=1}^M\rho_{m}^2P_{km}\sum_{k'=1}^{K}\beta_{k'm}^2P_{k'}+MG^2\sum_{m=1}^M \rho_{m}^2P_{km}\sigma^2+\sigma_{H}^2(\sum_{m=1}^M\sum_{k'=1}^{K}\beta_{k'm}^2P_{k'}+M\sigma^2)}.
\end{equation}
\normalsize
\end{figure*}

\end{proposition}
\begin{proof}
Please see Appendix \ref{proof_sinr}.
\end{proof}

We employed the well-known use-and-then-forget bound \cite{marzetta2016fundamentals}  to derive the rate expression presented in (\ref{rate-rev}). This technique involves performing match-filtering at each UxNB based on the estimated user-to-UxNB channels. Importantly, this method only requires the local knowledge of instantaneous CSI at the UxNBs, while the HAPS solely needs to know the large-scale fadings for the user-to-UxNBs channels.
In our paper, we assume perfect CSI knowledge at the UxNBs. 
At the second time slot, we perform the analog beamforming and combining for the second hops, i.e., UxNB-to-HAPS channels, utilizing the steering vectors of the transmit and receive antennas of the UxNBs and HAPS, respectively. As these vectors change slowly, we assume perfect knowledge of them at both the UxNBs and HAPS. During the decoding phase at the HAPS, the estimated channels for user-to-UxNBs are disregarded, and we decode the signal of each user utilizing steering vectors. Note that for deriving (\ref{rate-rev}), we considered that there is no channel estimation error.

Note that in the THz band, addressing the significant path loss requires the use of a massive number of antenna elements \cite{THz_6g_Precoding}. Implementing digital beamforming techniques would necessitate one RF chain per antenna element, resulting in substantial power consumption. To mitigate this issue, we adopted analog beamforming for transmitting signals from UxNBs to HAPS in the sub-THz band. This approach enables the utilization of just one RF chain per UxNB's transmitter, significantly reducing power consumption. Additionally, analog beamforming is employed at HAPS to combine the received signals from all UxNBs. This further minimizes power requirements by employing a single RF chain for all antenna elements at HAPS.

In terms of 5G terminology, the UxNB nodes in the system model in Fig. \ref{system-model} have the same functionality with distributed units (DUs), and the HAPS has the same functionality with a centralized unit (CU) \cite{ahmadi20195g}. In the study item for the new radio access technology, different functional
splits between the CU and the DU have been studied \cite{cu-du-split}.
Split option 7 itself has three variants, depending on what aspects of the
physical layer processing are performed in the DU and CU. 
Indeed, in option 7, the lower physical layer functions and
RF circuits are located in the DU(s), and the upper protocol layers including the upper
physical layer functions reside in the CU. Our proposed transceiver scheme at each UxNB in Fig. \ref{transc_signal} resides between the option 7-2 and option 7-3 functional split between the CU and DU. This is because we do not demodulate and decode the received signals at DU, and we simply apply match-filtering and power allocation for the received signals. Then, we upconvert the signals to the sub-THz band and transmit them to the CU (HAPS) in orthogonal RBs. It is noteworthy that all of these functions happen in layer 1 (physical layer). We should add that, in general, the links between the DUs and CU are assumed to be perfect when connected with a fiber link. However, in our proposed transceiver scheme for the uplink, at each DU, we have a transmission part to prepare and design a beamformer for the signals, and forward them to the CU wirelessly. This is one of the main differences between our work and terrestrial cell-free schemes in the literature.

\section{Optimization problem}
In the previous section, we derived a closed-form expression in \eqref{SINR} for the achievable rate of each user based on the proposed transceiver and receiver schemes. However, the values of the allocated powers and locations of the UAVs are unknown at this rate formula. In this section, we formulate an optimization problem that maximizes the minimum SINR of users by finding the optimum allocated powers for users in each UxNB (i.e., $\bold{P}=[P_{km}]_{K\times M}$) 
 and the optimum locations of UxNBs (i.e., $\bold{x}=[x_{d,m}]_{1\times M}$ and $\bold{y}=[y_{d,m}]_{1\times M}$). Note that maximizing the minimum SINR of users leads to fair achievable rates among all users. 
We can write the optimization problem as follows:
\small
\begin{alignat}{2}
\text{(P):~~~~    }&\underset{\bold{P},\bold{x},\bold{y}}{\text{max}}~~~        && \underset{k}{\text{min}}  ~~\mathsf{SINR_k}\label{eq:obj_fun}\\
&\text{s.t.} &      &   \sum_{k=1}^KP_{km}\leq P_m, ~~\forall m, \label{eq:constraint-sum}\\
               &&& P_{km}>0, ~~\forall k, ~\forall m, \label{eq:constraint+}\\&&& x_{\text{min}}\leq x_{d,m}\leq x_{\text{max}},~y_{\text{min}}\leq y_{d,m}\leq y_{\text{max}},~\forall m,\label{eq:constraint-range-p}
\end{alignat}
\normalsize
 where constraint (\ref{eq:constraint-sum}) shows the maximum total power constraint at UxNB $m$, and constraint (\ref{eq:constraint-range-p}) indicates the horizontal range of the UxNBs' flight. Problem (P) is a non-convex optimization problem due to its objective function, which is not concave with respect to variables $\bold{P}$, $\bold{x}$, and $\bold{y}$. In order to solve this problem, first, based on the BCD method \cite{razaviyayn2013unified}, we split the optimization problem into two sub-problems. In the first sub-problem, we solve the problem (P) for the case where locations of the UxNBs (i.e., $\bold{x}$ and $\bold{y}$) are given, and in the second optimization sub-problem, the power allocation coefficients $\bold{P}$ are assumed to be given. It is important to stress that the power allocation problem can be  equivalent to the case where we do not have control of the locations of the UxNBs, and they are non-dedicated aerial users that perform their own missions, and we utilize them as APs of the cell-free scheme as well. Both power allocation and deployment sub-problems are still non-convex, and we solve them by means of the bisection and SCA methods in the following two sub-sections, respectively.
 
 \subsection{Power Allocation Sub-Problem}
If we fix the locations of the UxNBs in (P), we will have the power allocation sub-problem as follows:
\small
\begin{alignat}{2}
\text{(P1):~~~~    }&\underset{\bold{P}}{\text{max}}~~~        && \underset{k}{\text{min}}  ~~\mathsf{SINR_k}\label{eq:obj_fun_P}\\
&\text{s.t.} &      &   \sum_{k=1}^KP_{km}\leq P_m, ~~\forall m, \label{eq:constraint-sum_P}\\
               &&& P_{km}>0, ~~\forall k, ~\forall m. \label{eq:constraint+_P}
\end{alignat}
\normalsize
This problem is still non-convex with respect to power allocation coefficients $\bold{P}$.
\begin{proposition}\label{propos_quasi}
The optimization problem (P1) is a quasi-concave optimization problem. 
\end{proposition}
\begin{proof}
Please see Appendix \ref{proof_quasi}.
\end{proof}

Since problem (P1) is a quasi-concave optimization problem, its optimal solution can be found efficiently by the bisection
method \cite{boyd2020disciplined}. To do this, we rewrite the problem (P1) by adding slack variable $\eta$ as follows:
\small
\begin{alignat}{2}
\text{(P2):~~~~    }&\underset{\bold{P},\eta}{\text{max}}~~~ \eta      && \label{eq:obj_fun_P}\\
&\text{s.t.} &      &  \mathsf{SINR_k}\geq \eta, ~\forall k \label{eq:constraint+eta-sinr} \\&&&\sum_{k=1}^KP_{km}\leq P_m, ~~\forall m, \label{eq:constraint-sum_P}\\
               &&& P_{km}>0, ~\eta >0, ~\forall k, ~\forall m. \label{eq:constraint+_P}
\end{alignat}
\small
\normalsize
It can be easily proven that (P2) is equivalent to (P1).
For this, it is noteworthy that the constraint (\ref{eq:constraint+eta-sinr}) at (P2) can be rewritten as $\min (\mathsf{SINR_1},...,\mathsf{SINR_K})\geq \eta$, and since we aim to maximize $\eta$, the optimal $\eta$ is given by $\eta=\min (\mathsf{SINR_1},...,\mathsf{SINR_K})$ which is the same with the objective function of (P1). We can also prove that at the optimal solution for (P2), we must have $\eta_{\mathrm{equi}}^*=\mathsf{SINR_1}=...=\mathsf{SINR_K}$. For this, consider that the SINR value of user $k1$ is greater than other users. This leads to a decrease in the SINR value of one other user $k2$ which is because of power constraint (\ref{eq:constraint-sum_P}) in (P2). Therefore, we can write $\mathsf{SINR_{k2}}\leq\mathsf{SINR_1}=...=\mathsf{SINR_K}\leq\mathsf{SINR_{k1}}$. This means that we have $\eta_{\mathrm{non-equi}}=\min (\mathsf{SINR_1},...,\mathsf{SINR_K})=\mathsf{SINR_{k2}}$. Obviously, this value is less than $\eta_{\mathrm{equi}}^*$, i.e., $\eta_{\mathrm{non-equi}}\leq\eta_{\mathrm{equi}}^*$, and therefore the proof is completed.
By performing the variable change $\bold{T}=[P_{km}^2]_{K\times M}$, for any given value of $\eta$, problem (P2) will be a convex feasibility
problem that can be solved optimally by standard convex optimization solvers. In this paper, we utilize CVX solver \cite{cvx} that applies the interior point method to solve these convex problems.

\subsection{UxNBs Placement Sub-Problem}
For a given allocated powers in (P), we will have the following placement sub-problem:
\small
\begin{alignat}{2}
\text{(P3):~~~~    }&\underset{\bold{x},\bold{y}}{\text{max}}~~~        && \underset{k}{\text{min}}  ~~\mathsf{SINR_k}\label{eq:obj_fun_P}\\
&\text{s.t.} &      &   x_{\text{min}}\leq x_{d,m}\leq x_{\text{max}},~y_{\text{min}}\leq y_{d,m}\leq y_{\text{max}},~\forall m.\label{eq:constraint-range}
\end{alignat}
\normalsize

This problem is still non-convex with respect to variables $\bold{x}$ and $\bold{y}$. By substituting the SINR formula from (\ref{SINR}) and introducing three slack variables $\eta$, $\bold{t}=[t_{k}]_{1\times K}$, and $\boldsymbol{\beta}=[\beta_{km}]_{K\times M}$, (P3) can be rewritten as follows:
\small
\begin{alignat}{2}
&&&\text{(P4):~~~~  }\underset{\bold{x},\bold{y},\eta,\bold{t},\boldsymbol{\beta}}{\text{max}}~~~         \eta \label{eq:obj_fun_P}\\
&&&\text{s.t.}   ~~~~ \frac{\sqrt{MG^2NSP_{k}}(\sum_{m=1}^M \gamma_{m}\sqrt{P_{km}} \beta_{km})}{t_k}\geq \sqrt{\eta}, ~\forall k, \label{sinr_constraint-p4}\\&&& ~~~~~~~~t_k^2 \geq MG^2\sum_{m=1}^M\rho_{m}^2P_{km}\sum_{k'=1}^{K}\beta_{k'm}^2P_{k'}+MG^2\sum_{m=1}^M \rho_{m}^2P_{km}\sigma^2 \label {t_con-p4}  \\&&& \nonumber~~~~~~~~~~~~~~~~~~~~~~~~~~~~~~+\sigma_{H}^2(\sum_{m=1}^M\sum_{k'=1}^{K}\beta_{k'm}^2P_{k'}+M\sigma^2), ~\forall k, \\&&&
 ~~~~~~~~\beta_{km}^{-1} \geq  \sqrt{(x_{u,k}-x_{d,m})^2+(y_{u,k}-y_{d,m})^2+(z_{d,m})^2} \label{beta_co-p4} \\&&&  ~~~~~~~~~~~~~~~~~~~~~~~~~~\nonumber\times 10^{\frac{P_{km}^{\mathsf{LoS}}\eta_{\mathsf{LoS}}^{\mathsf{dB}}+P_{km}^{\mathsf{NLoS}}\eta_{\mathsf{NLoS}}^{\mathsf{dB}}}{20}}\beta_0^{-\frac{1}{2}}, ~\forall k, ~\forall m, 
\\
&&&  ~~~~~~~~x_{\text{min}}\leq x_{d,m}\leq x_{\text{max}},~y_{\text{min}}\leq y_{d,m}\leq y_{\text{max}},~\forall m, \label{eq:constraint-range-p4}\\&&&
 ~~~~~~~~ \eta > 0, ~t_k > 0,~ \beta_{km} > 0, ~\forall k, ~\forall m.
\label{eq:constraint-pos-p4}
\end{alignat}
\normalsize
The equivalency proof of (P4) with (P3) is similar to the one for (P2) and (P1). Indeed, the constraint (\ref{sinr_constraint-p4}) at (P4) can be written as $\mathsf{SINR_k}\geq \eta, \forall k$, that is equivalent to $\min (\mathsf{SINR_1},...,\mathsf{SINR_K})\geq \eta$. Since we aim to maximize $\eta$, the optimal $\eta$ is given by $\eta=\min (\mathsf{SINR_1},...,\mathsf{SINR_K})$ which is the same with the objective function of (P3). Note that the numerator of the left-hand side of (\ref{sinr_constraint-p4}) is the square of the numerator of the SINR formula in (\ref{SINR}). Also, the slack variable $\bold{t}=[t_{k}]_{1\times K}$ in constraint (\ref{t_con-p4}) indicated the square of the denominator of the SINR formula in (\ref{SINR}). The slack variable $\boldsymbol{\beta}=[\beta_{km}]_{K\times M}$ in constraint (\ref{beta_co-p4}) shows the large-scale channel gain in (\ref{bkm_formula}) for the link between user $k$ and UxNB $m$. By taking the logarithm of both sides of constraint (\ref{t_con-p4}), it will be in the form of a norm function of variable $\boldsymbol{\beta}$ less than an affine function of variable $\bold{t}$, and hence it is a convex set. However, this problem is still non-convex due to constraints (\ref{sinr_constraint-p4}) and (\ref{beta_co-p4}). To address constraint (\ref{sinr_constraint-p4}), we perform a variable change such that $\eta=\zeta^4$. Hence, for the optimization problem (P4) we can write the following:
\small
\begin{alignat}{2}
\text{(P5):~~~~  }&\underset{\bold{x},\bold{y},\zeta,\bold{t},\boldsymbol{\beta}}{\text{max}} ~~~~       && \zeta \label{eq:obj_fun_P}\\
&\text{s.t.}   && \frac{1}{t_k}\geq \frac{\zeta^2}{\sqrt{MG^2NSP_{k}}(\sum_{m=1}^M \gamma_{m}\sqrt{P_{km}} \beta_{km})}, ~\forall k, \label{sinr_constraint_zeta}\\&&& (\ref{t_con-p4}),~ (\ref{beta_co-p4}), ~(\ref{eq:constraint-range-p4}), ~(\ref{eq:constraint-pos-p4}).
\label{eq:constraints-p4}
\end{alignat}
\normalsize
One can see that the right-hand side of (\ref{sinr_constraint_zeta}) is in the form of a quadratic function of variable $\zeta$ over an affine function of variable $\boldsymbol{\beta}$, which is known to be a convex function when its denominator is positive \cite{boyd2004convex}. Also, the right-hand side of (\ref{beta_co-p4}) is in the form of a norm function of variables $\bold{x}$ and $\bold{y}$, and so it is convex. However, the left-hand side of the constraints (\ref{beta_co-p4}) and (\ref{sinr_constraint_zeta}) (i.e., $\frac{1}{\beta_{km}}$ and $\frac{1}{t_{k}}$) are not in the form of a concave function. In order to manage these non-concave terms, we propose an iterative scheme based on the SCA method \cite{Traj_Omid}. In this method, the original non-convex problem is optimized by solving convex approximations of the original problem iteratively around an initial point until convergence. 
We know that the first-order Taylor series expansion of a convex function $f(z)$ provides a global lower-bound for that function, i.e., $f(z)\geq f(z_0)+\nabla f(z_0)^T(z-z_0)$.
Given that the left-hand side of the constraints (\ref{beta_co-p4}) and (\ref{sinr_constraint_zeta}) are in the form of a convex function, we can approximate them at iteration $l+1$ by their first-order Taylor series expansion around the solution of the previous iteration $l$. Therefore, at the iteration ($l+1$), we replace the  left-hand side of the constraints (\ref{beta_co-p4}) and (\ref{sinr_constraint_zeta}) with the expressions
      $-\frac{1}{(\beta_{km}^l)^2}(\beta_{km}-\beta_{km}^{l})+\frac{1}{\beta_{km}^l}$ and $ -\frac{1}{(t_{k}^l)^2}(t_{k}-t_{k}^{l})+\frac{1}{t_{k}^l}$, respectively. It is clear that these functions are affine functions with respect to variables $\bold{t}$ and $\boldsymbol{\beta}$. Now, by replacing these approximations in (P5), the optimization problem at the iteration $l+1$ of the SCA method around the initial points $\bold{t}^{l}$ and $\boldsymbol{\beta^{l}}$ is given by
      
  \small    
      \begin{alignat}{2}
&&&\text{(P6):~~~~  }\underset{\bold{x},\bold{y},\zeta,\bold{t},\boldsymbol{\beta}}{\text{max}}~~~         \zeta \label{eq:obj_fun_P}\\
&&&\text{s.t.}   ~~~~ -\frac{1}{(t_{k}^l)^2}(t_{k}-t_{k}^{l})+\frac{1}{t_{k}^l}\geq \frac{\zeta^2}{\sqrt{MG^2NSP_{k}}(\sum_{m=1}^M \gamma_{m}\sqrt{P_{km}} \beta_{km})}, ~\forall k, \label{sinr_constraint_zeta_2}\\&&&  ~~~~~~~~
-\frac{1}{(\beta_{km}^l)^2}(\beta_{km}-\beta_{km}^{l})+\frac{1}{\beta_{km}^l} \geq 10^{\frac{P_{km}^{\mathsf{LoS}}\eta_{\mathsf{LoS}}^{\mathsf{dB}}+P_{km}^{\mathsf{NLoS}}\eta_{\mathsf{NLoS}}^{\mathsf{dB}}}{20}} \label{beta_co} \\&&& \nonumber  ~~~~~~~~~~~~ \times \beta_0^{-\frac{1}{2}}\sqrt{(x_{u,k}-x_{d,m})^2+(y_{u,k}-y_{d,m})^2+(z_{d,m})^2}, ~\forall k, ~\forall m, 
\\
&&&  ~~~~~~~~(\ref{t_con-p4}), ~(\ref{eq:constraint-range-p4}), ~(\ref{eq:constraint-pos-p4}).
\label{eq:constraint-pos}
\end{alignat}
\normalsize

This optimization problem is convex since all of the constraints are in the form of a convex function of variables less than a concave function of variables. This problem can be efficiently solved by convex optimization techniques, such as the interior-point method. The optimum value for the objective function of (P4) can be derived from $\eta=\zeta^{\frac{1}{4}}$. 


\subsection{Iterative Algorithm for Joint Power Allocation and Placements of UxNBs}
In this subsection, we apply the BCD method \cite{razaviyayn2013unified} to solve the original problem (P), and find optimized values for power allocation and location variables. For this, we solve the placement of UxNBs and power allocation sub-problems alternately. First, we initialize the power allocation variables and the locations of UxNBs. Then, at each iteration, we optimize the location variables utilizing SCA method. Next, utilizing these locations, we find the optimum power allocation with the aid of the bisection method. These powers are used as initial values in the next iterations. The associated algorithm is summarized in Algorithm 1. 

Now, we prove the convergence of Algorithm 1. At iteration $l$ of Algorithm 1, problems (P6) and (P2) are solved. The minimum SINRs of users at iteration $l$ of problems (P6) and (P2)  are shown with $\eta_{\mathrm{P6,lb}}^l$ and $\eta_{\mathrm{P2}}^l$, respectively. Note that in constraint (\ref{sinr_constraint_zeta_2}) of (P6), we applied the first-order Taylor expansion for the approximation of a convex function. Since the Taylor series expansion of a convex function provides a global lower bound for that function, for the left-hand side of constraint (\ref{sinr_constraint_zeta_2}), we can write $\frac{1}{t_{k}^{l+1}}\geq-\frac{1}{(t_{k}^l)^2}(t_{k}^{l+1}-t_{k}^{l})+\frac{1}{t_{k}^l}$. Therefore, the minimum SINR of the original problem (P), which is derived by replacing the optimized locations in (P6) and is shown with $\eta_{P,\mathrm{location}}^{l+1}$, is lower-bounded by the approximated minimum SINR in (P6), i.e., $\eta_{\mathrm{P6,lb}}^{l+1}$. In other words, we have $\eta_{P,\mathrm{location}}^{l+1}\geq\eta_{\mathrm{P6,lb}}^{l+1}$. Moreover, since we maximize  the objective function of (P6) at iteration $l+1$, we can write $\eta_{\mathrm{P6,lb}}^{l+1}\geq\eta_{\mathrm{P6,lb}}^l$. Considering this point that the Taylor expansion of a function around an initial point is exactly equal to the value of the original function at that point, i.e., $\eta_{P,\mathrm{location}}^{l}=\eta_{\mathrm{P6,lb}}^l$, we can conclude that $\eta_{P,\mathrm{location}}^{l+1}\geq\eta_{P,\mathrm{location}}^{l}$, which means that the objective function of the original problem (P) is non-decreasing with $l$ for location problem (P6). For the power allocation problem (P2), because there is no approximation for SINR of users,  the minimum SINR of the original problem (P), which is derived by replacing the optimized powers in (P2) and is shown with $\eta_{P,\mathrm{power}}^{l+1}$, is equal to the minimum SINR in (P2), i.e., $\eta_{\mathrm{P2}}^{l+1}$. In other words, we have $\eta_{P,\mathrm{power}}^{l+1}=\eta_{\mathrm{P2}}^{l+1}$. Since the bisection method for solving (P2) gives the optimal values for power variable, we can write $\eta_{P,\mathrm{power}}^{l+1}\geq\eta_{P,\mathrm{power}}^{l}$, which means that the objective function of the original problem (P) is non-decreasing with $l$ for power problem (P2). Since the Taylor series expansion of a convex function provides a global lower bound for that function, the proposed sub-optimal algorithm for (P) is upper-bounded by the optimal solution of (P) which is shown by $\eta_{P}^*$. Therefore, we can write $...\geq\eta_{P,\mathrm{location}}^{l}\geq\eta_{P,\mathrm{power}}^{l}\geq\eta_{P,\mathrm{location}}^{l+1}\geq\eta_{P,\mathrm{power}}^{l+1}\geq...\geq\eta_{P}^*$. Because the proposed sub-optimal algorithm to solve (P) leads to a non-decreasing objective function over iterations and is globally upper-bounded, it is guaranteed to converge.

\begin{algorithm}
\caption{Iterative algorithm based on BCD to jointly find the power allocation and placement of UxNBs in problem P.}\label{alg1}
\begin{algorithmic}[1]
\State  Initialize power allocation variables $\bold{P}^{l}$, and UxNB locations $\bold{x}^{l}$ and $\bold{y}^{l}$, and let $l=0$. 
\State Initialize the values of $\eta_{\mathsf{min}}$ and $\eta_{\mathsf{max}}$ for the bisection method, where $\eta_{\mathsf{min}}$ and $\eta_{\mathsf{max}}$ show a range for the minimum $\mathsf{SINR}$ of users. Choose a tolerance $\epsilon>0$.
\Repeat
\State Solve the convex placement problem (P6) with the given powers $\bold{P}^{l}$ and initial locations $\bold{x}^{l}$ and $\bold{y}^{l}$ by the interior-point method (utilizing the CVX solver \cite{cvx}) and find the optimum values for  $\mathbf{x}$ and $\mathbf{y}$.
\Repeat
\State Set $\eta=\frac{\eta_{\mathsf{min}} +\eta_{\mathsf{max}}}{2}$.
\State Solve the convex feasibility
problem (P2) by the interior-point method (utilizing the CVX solver \cite{cvx}) with the given UxNBs locations $\bold{x}$ and $\bold{y}$, and find the optimal values for $\bold{P}$.
\State If the problem (P2) is feasible, set $\eta_{\mathsf{min}}=\eta$; else set $\eta_{\mathsf{max}}=\eta$.
\Until {$\eta_{\mathsf{max}} -\eta_{\mathsf{min}}<\epsilon$.}
\State Update $l=l+1$; and set $\mathbf{x}^l=\mathbf{x}$, $\mathbf{y}^l=\mathbf{y}$, and $\mathbf{P}^l=\mathbf{P}$.
\Until {Convergence or a maximum number of iterations is reached.}
\end{algorithmic}
\end{algorithm}

\subsection{Computational Complexity Analysis}
Now, the computational complexity of the proposed iterative solution for problem (P) in Algorithm 1 is presented. At  iteration $l$ of the BCD method in Algorithm 1, computational complexity is dominated by solving convex problems in (P2) and (P6). These convex problems are solved using the interior point method. According to \cite{boyd2004convex}, the interior point method
 requires $\log(\frac{n_c}{t^0\varrho})/\log\varepsilon$ number of iterations
(Newton steps) to solve a convex problem, where $n_c$ is the total
number of constraints, $t^0$ is the initial point for approximating the accuracy of the interior-point
method, $0<\varrho\ll1$ is
the stopping criterion, and $\varepsilon$ is used for updating the accuracy
of the interior point method. For (P2) and (P6), number of constraints are equal to $n_c^p=2KM+K+M$ and $n_c^x=2KM+3K+4M+1$, respectively. Note that superscripts $p$ (for power) and $x$ (for location) refer to problems (P2) and (P6), respectively. Therefore, the computational complexity of Algorithm 1 will be  $\mathcal{O}\Big(N_{\mathrm{iter}}(\frac{\log(\frac{KM}{t^{0,x}\varrho^{x}})}{\log\varepsilon^{x}}+N_{\mathrm{iter}}^b \frac{\log(\frac{KM}{t^{0,p}\varrho^{p}})}{\log\varepsilon^{p}})\Big)$, where $N_{\mathrm{iter}}$ and $N_{\mathrm{iter}}^b$ indicate the number of iterations for convergence of Algorithm 1 and the bisection method, respectively. Note that the optimization problem (P) formulated in this paper is non-convex with respect to its power and location variables, and hence the complexity of finding the optimal values for its variables is NP-hard \cite{Traj_Omid}. However, at each iteration of the proposed method in Algorithm 1, computational complexity is dominated by solving two convex problems. Therefore, the proposed method in Algorithm 1 is a low-complexity solution that finds optimized values for power and location variables.

Note that the LoS channel model for UxNB to HAPS links can be easily generalized for the NLoS case. In scenarios where the UxNB to HAPS links operate in lower frequency bands or UxNBs fly at lower altitudes, we can simply utilize the same probabilistic channel model in (\ref{prob_hkmn}) and \eqref{bkm_formula} for users to UxNB links. In this case, since the LoS path dominates the NLoS paths, we can continue employing the same analog beamforming coefficients that we applied in the LoS case for UxNB transmit antennas. Additionally, the channel hardening benefit offered by massive MIMO \cite{emil} enables the utilization of only large-scale fading for calculating the achievable rates of the UxNB to HAPS links which allows us to employ the same SINR formula as in \eqref{SINR} for the NLoS scenario. Consequently, the optimization problem defined in (P) and its corresponding solution in Algorithm 1 can be applied to this NLoS scenario as well.

 Finally, it is important to highlight that the optimization problem (P) presented in the paper aims to find the optimized locations of UAVs and power allocations within a single time slot. However, it is worth noting that this problem, along with the corresponding solution described in Algorithm 1, can be readily extended to accommodate scenarios where UAVs are mobile across multiple time slots. In such cases, the objective function can be modified to minimize the maximum SINR of all users over the entire duration of $T$ time slots. Additionally, to account for UAV movement constraints during each time slot, an additional convex constraint can be introduced. By incorporating these modifications, the proposed method outlined in Algorithm 1 can effectively determine the allocated powers for each user at each UxNB over all time slots and optimize the trajectory of all UAVs across the $T$ time slots.

\section{Numerical Results}
In this section, numerical results are provided in order to show the performance gain of the proposed scheme. The following default parameters are applied in the simulations except that we specify different values for them. For the carrier frequency of the first and second hops, we assume $f_\mathsf{sub6}=2~\mathrm{GHz}$ and $f_\mathsf{THz}=120~\mathrm{GHz}$, respectively. The communication bandwidth is assumed to be $\mathsf{BW}=1~\mathrm{MHz}$, and the noise power spectral density is $N_0=-174~\mathrm{dBm/Hz}$. We assume that all users are uniformly distributed over a square area with a length of $1000$ meters. Also, we assume that the HAPS is deployed in the middle of this square area. The default values for the number of users, number of UxNBs, number of antenna elements in the receive UPA of each UxNB, number of antenna elements in the transmit UPA of each UxNB, and number of antenna elements in the receive UPA of the HAPS are equal to $K=16$, $M=16$, $N=4$, $G=9$, and $S=400$, respectively. 
 We set the maximum transmit power at each user and each UxNB as $P_k=0.2~W,~\forall k$ and $P_m=25 ~\mathrm{dBm},~\forall m$, respectively. We assume that all users send their signals with their maximum power. Considering an urban area, the excessive path loss affecting the air-to-ground links in LoS and NLoS cases is assumed to be $\eta_{\mathsf{LoS}}^{\mathsf{dB}}=1~\mathsf{dB}$ and $\eta_{\mathsf{NLoS}}^{\mathsf{dB}}=20~\mathsf{dB}$, respectively \cite{Irem}. Also, for the urban area, we have $A=9.61$ and $B=0.16$ \cite{Irem}.  The absorption coefficient of the sub-THz medium for $f_\mathsf{THz}=120~\mathrm{GHz}$ is equal to $K=0.5~\mathrm{dB/km}$ and the effective height is given by $h^{e}=1.6~ \mathrm{km}$ \cite{grace2011broadband}. We assume an initial uniform square placement for $M$ UAVs,
 and a fixed flight height $z_{d,m}=120~\mathrm{m},~\forall m$ for all of the UAVs. In Algorithm 1, we set initial values as $\eta_{\mathsf{min}}=0$ and $\eta_{\mathsf{max}}=1500$, and also the tolerance value as $\epsilon=0.01$. We also consider the HAPS altitude as 20 km. 
 
 In the simulation figures, the proposed scheme, referred to as the \textbf{aerial cell-free scheme}, is compared with two well-known baseline schemes. \textbf{1) Aerial cellular scheme} in which each terrestrial user is served by only one UxNB \cite{Elham_conf, Irem}. All other parameters and channel models are the same as those in the proposed scheme. The backhauling of this baseline scheme is also conducted through the HAPS in the sub-THz band. \textbf{2) Terrestrial cell-free scheme} in which each terrestrial user is served by multiple terrestrial access points, and a perfect backhaul with fiber links is assumed to connect the APs and CPU \cite{ngo, Manijeh}. In this baseline scheme, it is assumed that there is no LoS link between terrestrial users and APs, so a Rayleigh fading model is considered for these channels. The path loss exponent is assumed to be 2 and 3.7 for LoS and NLoS links, respectively. Furthermore, the optimization algorithms applied in the baseline schemes are similar to the one used for the proposed scheme, except for the differences in the channel models for the user-to-AP and AP-to-HAPS links. To ensure a fair comparison, the number of receiver antenna elements at terrestrial access points is assumed to be the same as the number of receiver antenna elements at UxNBs in aerial schemes (i.e., $N$). The match-filtering method is applied in each access point in both baseline schemes to align the received signals.
 
\begin{figure}[t]
\centering
  \includegraphics[width=\linewidth,height=5 cm]{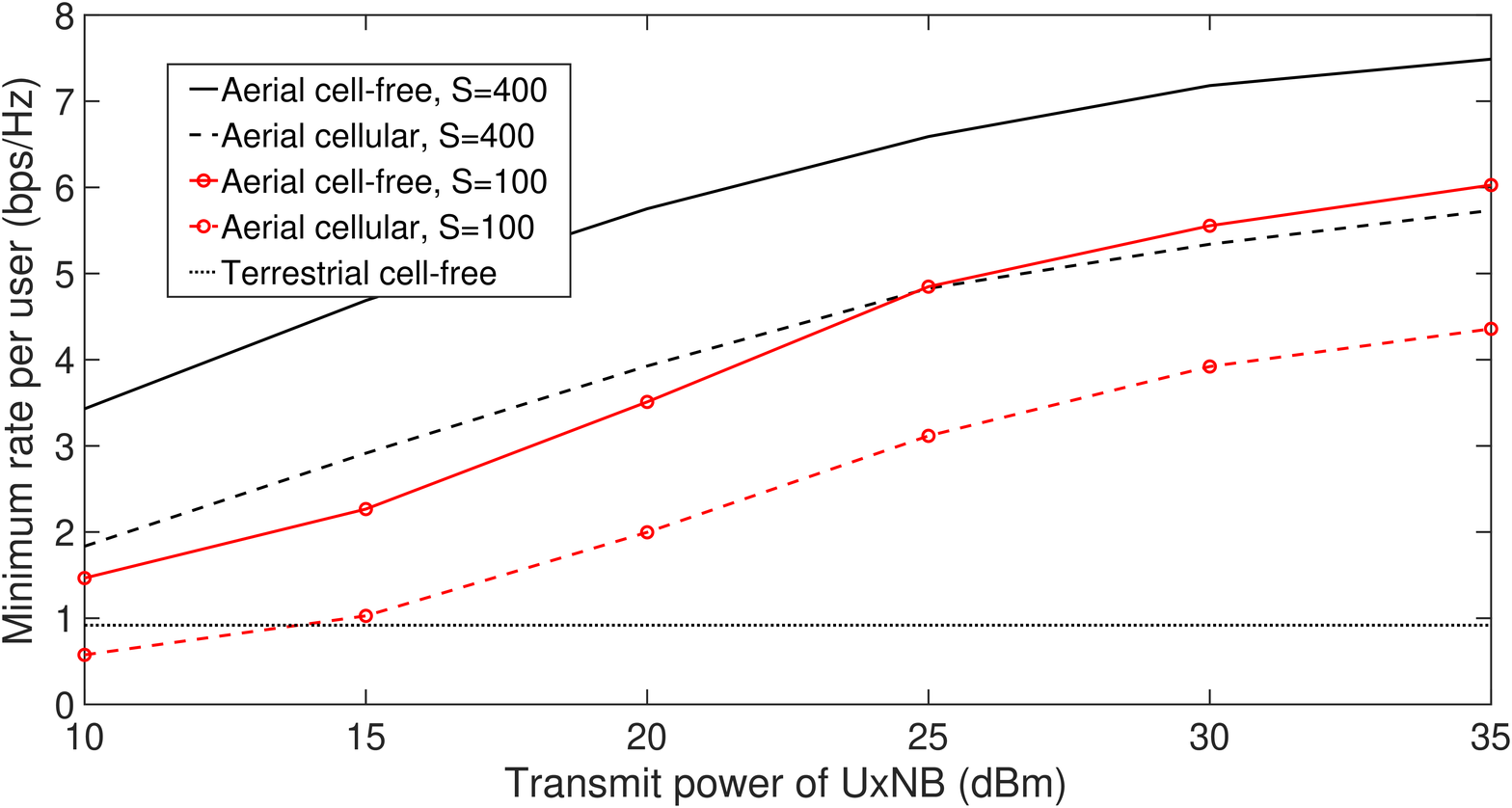}
  \captionof{figure}{The achievable minimum rate per user versus the total power of each UxNB ($P_1=...=P_M$) for the aerial and terrestrial BSs with cell-free and cellular schemes. We set $K=16$, $M=16$, $N=4$, and $G=9$.}
  \label{R_vs_P_K16_M16_S400_100}
\end{figure}

Fig. \ref{R_vs_P_K16_M16_S400_100} shows the achievable minimum rate per user versus the total power of each UxNB ($P_1=...=P_M$) for the aerial and terrestrial BSs with cell-free and cellular schemes. 
We can see that the proposed aerial cell-free scheme has a better performance compared with the aerial cellular scheme for both values of $S=100$ and $S=400$. Indeed, due to the severe intercell interference from other users in the neighboring cells in the aerial cellular scheme, our proposed scheme has much better performance than this baseline scheme. Also in this figure, we can see that the terrestrial cell-free baseline scheme has a constant value by increasing APs powers since we considered a perfect backhaul for this scheme. Our proposed scheme has much better performance than this baseline scheme as well. This is because in the terrestrial cell-free scheme, due to a high path loss and shadowing, the link between a user and far access points can be very weak, and hence working in cell-free mode is not useful. However, in the proposed aerial cell-free scheme, since there is a strong LoS link between the users and UxNBs, the signal of each user is received at multiple UxNBs, and so the cell-free scheme is useful for the proposed system model.

\begin{figure}
\centering
  \includegraphics[width=\linewidth,height=5cm]{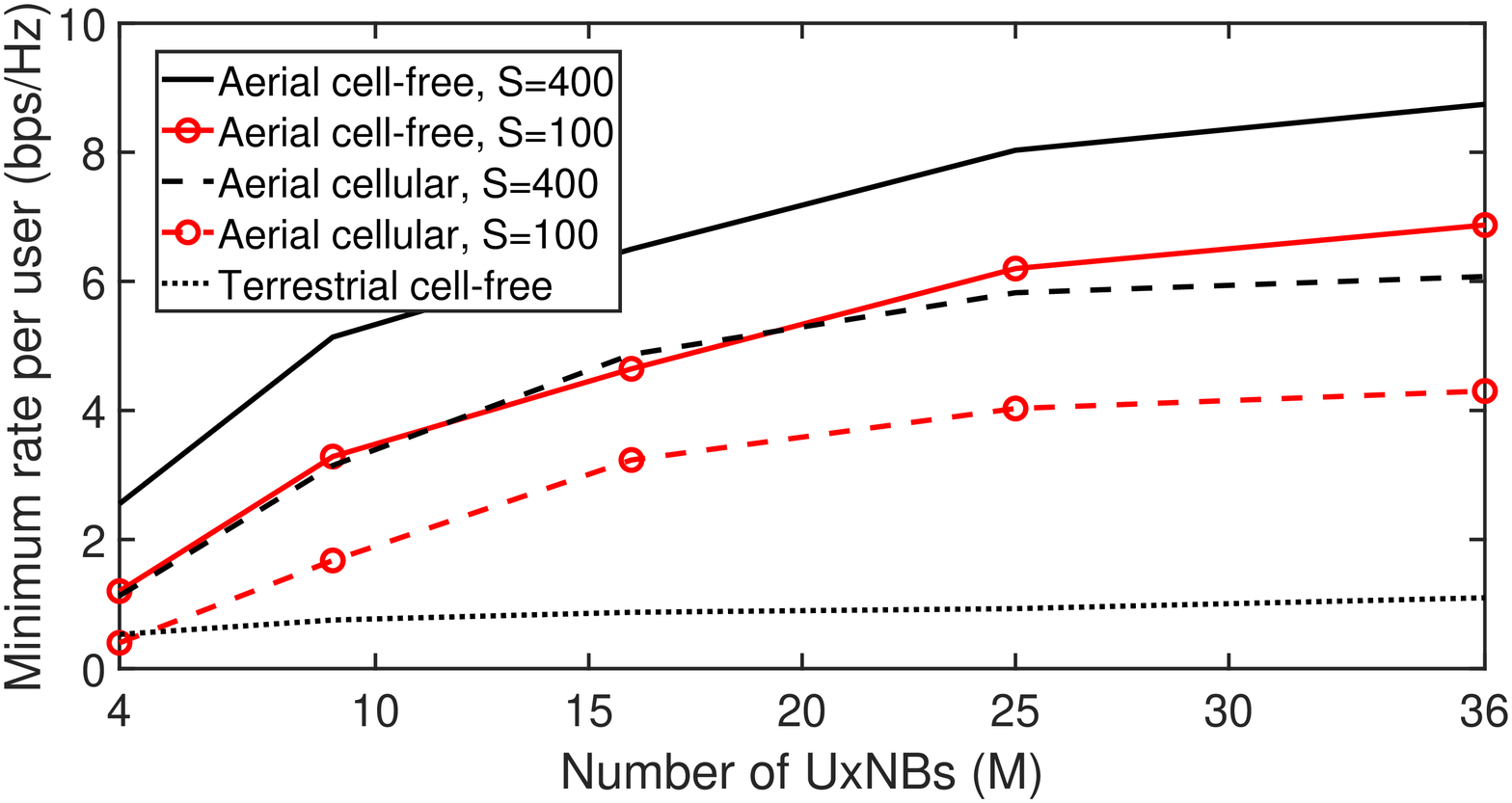}
  \captionof{figure}{The achievable minimum rate per user versus the number of UxNBs ($M$) for the aerial and terrestrial BSs with cell-free and cellular schemes. We set $P_m=25~\mathrm{dBm},~\forall m$, $K=16$, $N=4$, and $G=9$.}
  \label{R_vs_M_K16_P25}
\end{figure}

Fig. \ref{R_vs_M_K16_P25} shows the achievable minimum rate per user versus the number of UxNBs ($M$) for the aerial and terrestrial BSs with cell-free and cellular schemes. 
 We can see that the proposed aerial cell-free scheme outperforms the aerial cellular scheme for both values of $S=100$ and $S=400$. 
Further, we can see that by increasing $M$, the improvement in the performance of aerial schemes is much more than the terrestrial scheme, and this is because of a higher probability of establishing LoS links between UxNBs and users for a larger $M$ in aerial schemes. Finally, it is also shown that the superiority of the aerial cell-free scheme over the aerial cellular scheme increases by $M$ which is due to a higher intercell interference for the cellular scheme at a higher $M$.

\begin{figure}
\centering
  \includegraphics[width=\linewidth,height=5cm]{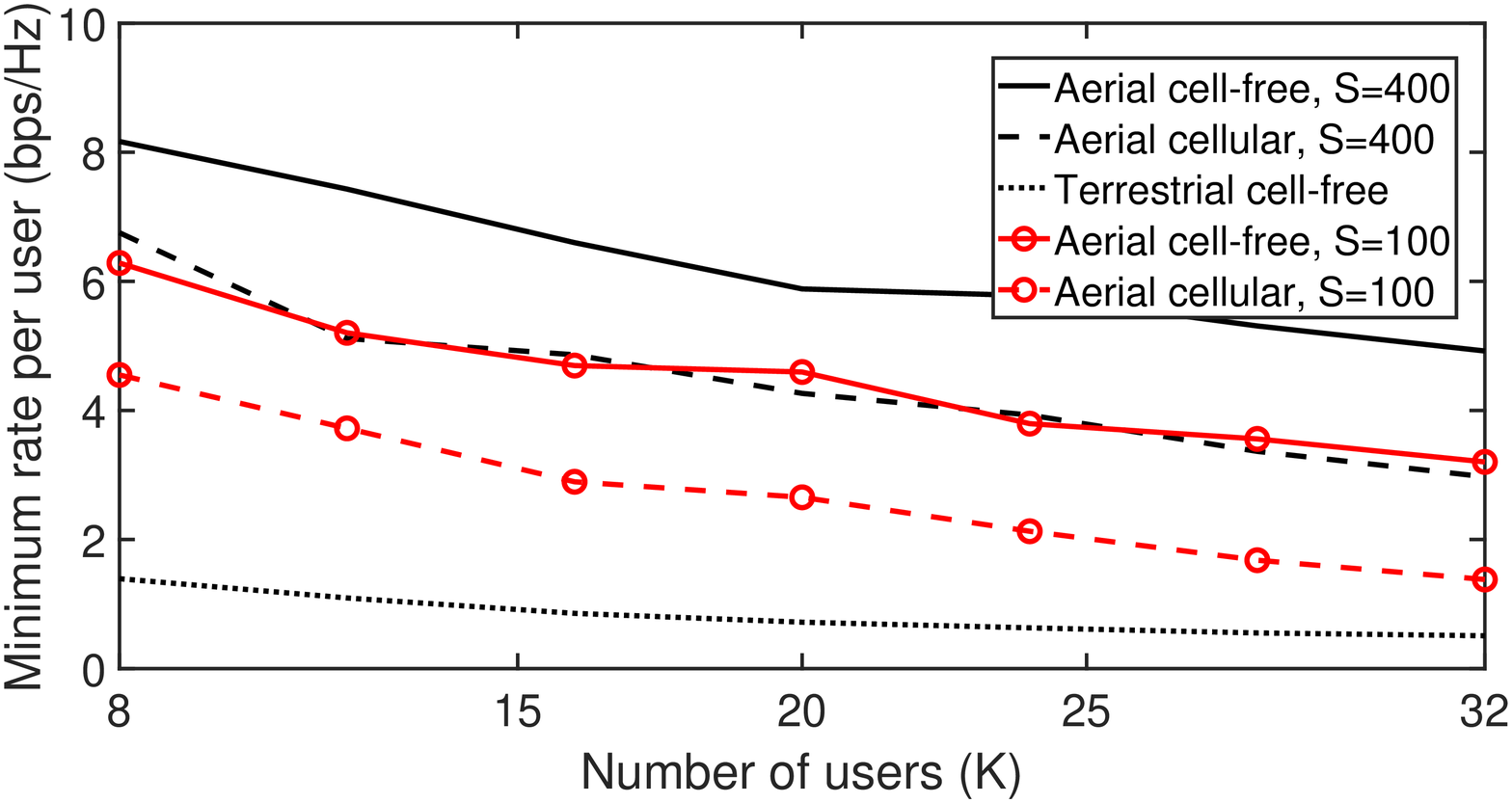}
  \captionof{figure}{The achievable minimum rate per user versus the number of users ($K$) for the aerial and terrestrial BSs with cell-free and cellular schemes. We set $P_m=25~\mathrm{dBm},~\forall m$, $M=16$, $N=4$, and $G=9$.}
  \label{R_VS_k_M16_P25_S400}
\end{figure}

Fig. \ref{R_VS_k_M16_P25_S400} indicates the achievable minimum rate per user versus the number of users ($K$) for the aerial and terrestrial BSs with cell-free and cellular schemes.
We set the parameter values as $P_m=25~\mathrm{dBm},~\forall m$, $M=16$, $N=4$, and $G=9$. 
As we can see, by increasing $K$, the performance of all schemes decreases. Also, we can see that the proposed aerial cell-free scheme performs better than both the aerial cellular and terrestrial cell-free baseline schemes due to the LoS link between the users and UxNBs. Further, by increasing $S$, the performance of the aerial schemes improves, which means that by increasing $S$ we can serve more users for a given minimum rate per user.

\begin{figure}
\centering
  \includegraphics[width=\linewidth,height=5cm]{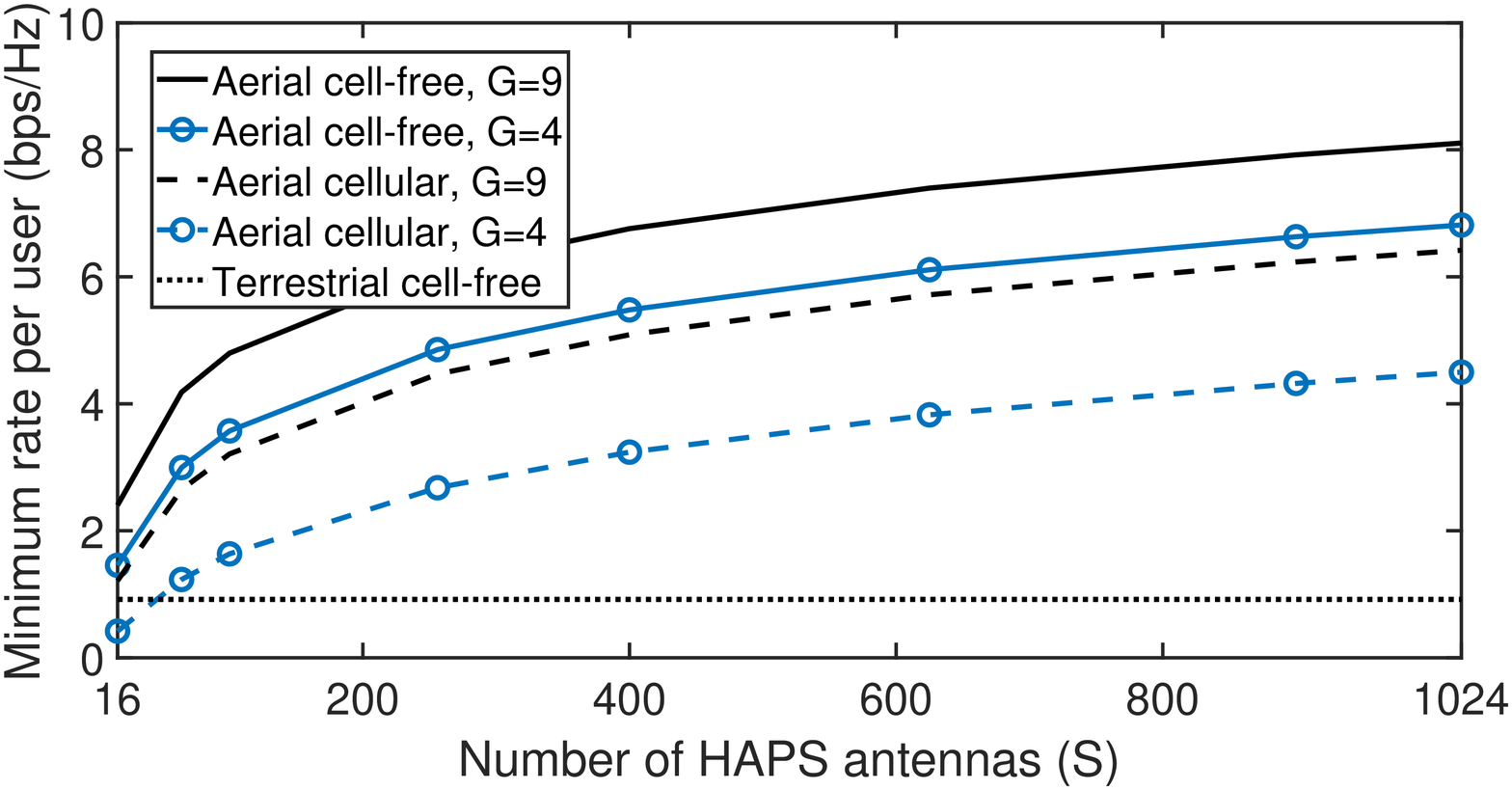}
  \captionof{figure}{The achievable minimum rate per user versus the number of HAPS antenna elements ($S$) for the cell-free and cellular schemes. We set $P_m=25~\mathrm{dBm},~\forall m$, $M=16$,  $K=16$, $N=4$, and $G=4,9$.}
  \label{R_vs_S_P25_M16_K16}
\end{figure}

Fig. \ref{R_vs_S_P25_M16_K16} indicates the achievable minimum rate per user versus the number of HAPS antenna elements ($S$). One can see that when the number of HAPS antenna elements is low, the performance of the aerial and terrestrial schemes are close. For example, when $S=16$, the terrestrial cell-free scheme performs better than the aerial cellular scheme, and its performance is comparable to the proposed aerial cell-free scheme. However, when the number of HAPS antenna elements is high, both aerial schemes have significant performance gain over the terrestrial cell-free scheme. This figure shows that utilizing a HAPS as a CPU is useful when the enormous path loss between the UxNBs and the HAPS in the sub-THz band is compensated for by a high number of antenna elements at the HAPS.

\begin{figure}[t]
\centering
  \includegraphics[width=\linewidth,height=5cm]{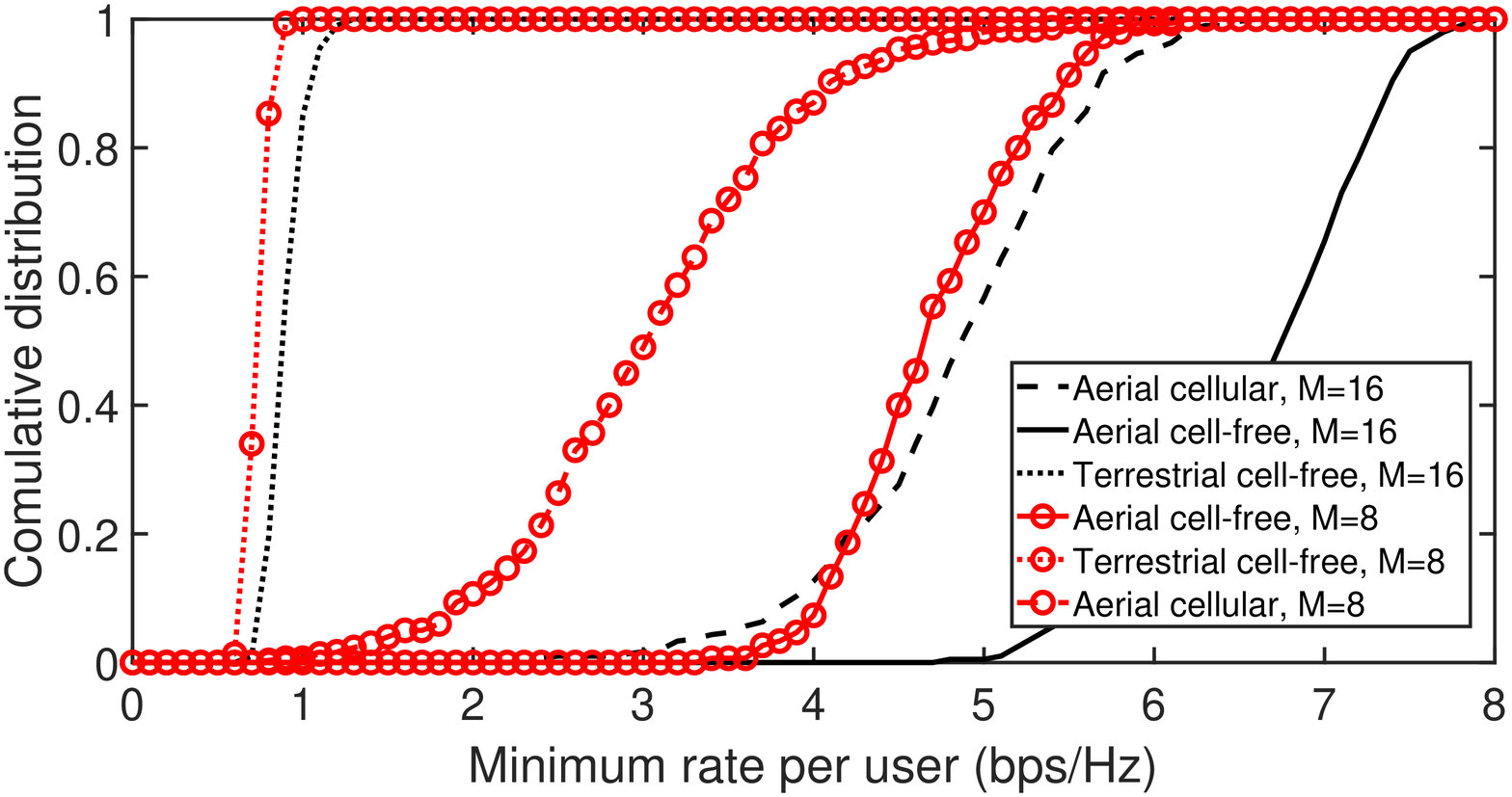}
  \captionof{figure}{The CDF of the minimum rate per user ($R_1=...=R_K$) for the aerial and terrestrial BSs with cell-free and cellular schemes. We set $P_m=25~\mathrm{dBm},~\forall m$, $K=16$, $S=400$, $N=4$, and $G=9$.}
  \label{CDF_K16_M16_S400}
\end{figure}

Fig. \ref{CDF_K16_M16_S400} shows the CDF of the minimum rate per user ($R_1=...=R_K$) for the aerial and terrestrial BSs with cell-free and cellular schemes.
We can see that the proposed aerial cell-free scheme has a better performance compared with both the aerial cellular and terrestrial cell-free baseline schemes for both values of $M=8$ and $M=16$. We can also see that the aerial cellular scheme outperforms the terrestrial cell-free scheme, and this is due to the LoS links that are established between users and UxNBs in the aerial networks. Further, we can see that the variance of the minimum achievable rate per user for the cell-free scheme is less than the cellular one, and we can observe that increasing the number of UxNBs reduces the variance of the minimum achievable rate per user for both aerial cell-free and aerial cellular schemes.

\begin{figure}
    \centering
  \includegraphics[width=\linewidth,height=5cm]{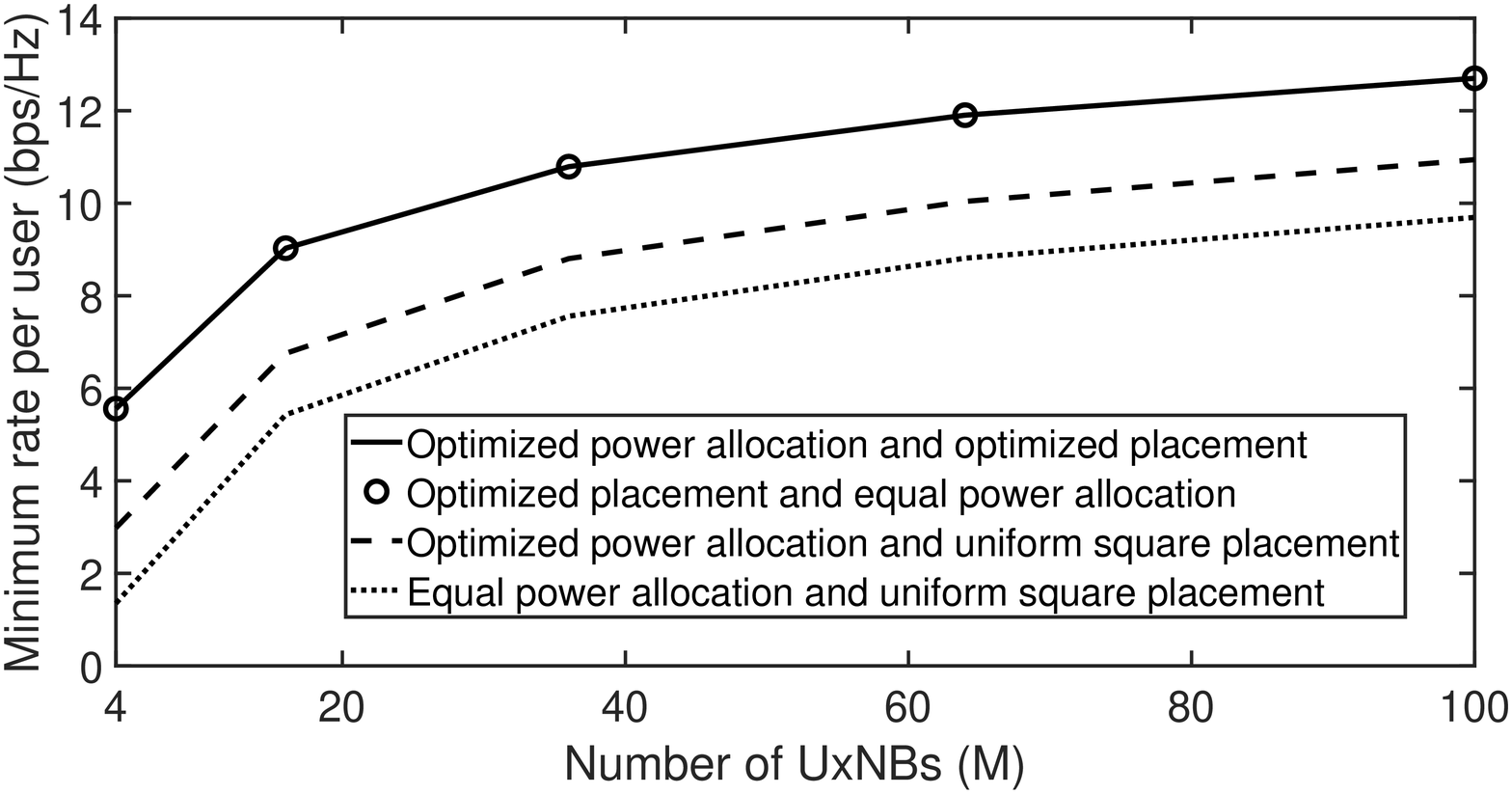}
  \captionof{figure}{The achievable minimum rate per user in the proposed cell-free scheme versus the number of UxNBs ($M$) for the cases where the locations of UxNBs and/or allocated power for each user at each UxNB are optimized. We assumed that $K=16$, $P_m=25~\mathrm{dBm},~\forall m$, $S=400$, $N=4$, and $G=9$.}
  \label{R_vs_M_PXY_K_16_P25_S400}
\end{figure}

Fig. \ref{R_vs_M_PXY_K_16_P25_S400} indicates the achievable minimum rate per user in the proposed cell-free scheme versus the number of UxNBs ($M$).
 In this figure, we explore heuristic baseline solutions for the allocated power of each user at each UxNB and the locations of UxNBs. In the heuristic solution for power allocation, we assume that the power of each UxNB is equally divided among users. For the heuristic solution regarding the locations of UxNBs, we uniformly deploy them over a square area. Notably, our proposed scheme is compared with three baseline solutions in this figure. In the first baseline scheme, we utilize heuristic solutions for both location and power variables. In the second baseline solution, we apply the heuristic solution for power and optimized values for locations. In the third baseline scheme, we utilize the heuristic solution for locations and optimized values for powers. We can see that when we use heuristic solutions for both power allocation and locations, the worst performance is achieved. Also, we can see that the optimization of power improves the performance of the proposed scheme. This figure indicates that the performance of the case where both power and location are jointly optimized with Algorithm 1 is the same as the performance of the case where the location is optimized with the SCA method and the heuristic solution is utilized for power allocation. This means that when the locations of UxNBs are optimized, equal power allocation is the optimal solution for power. This is because by optimizing the locations of the UxNBs, all users reach equal rates, and hence the optimal policy for power allocation is to divide the total power equally among the users. Please note that at the optimal point of our max-min optimization problem, all users must have the same rates.


 \begin{figure}
     \centering
  \includegraphics[width=\linewidth,height=5cm]{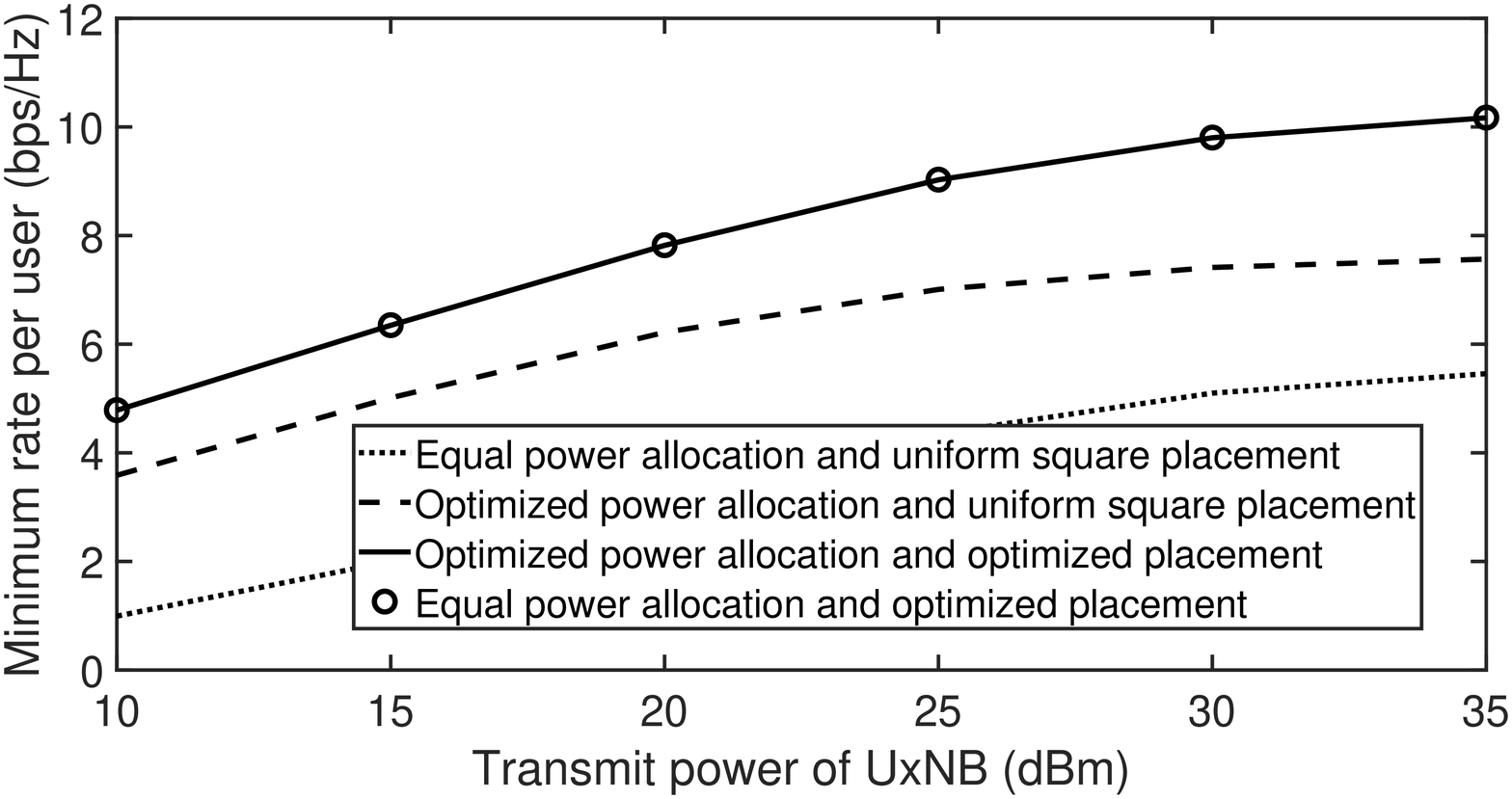}
  \captionof{figure}{The achievable minimum rate per user in the proposed cell-free scheme versus total power of each UxNB ($P_1=...=P_M$), for the cases where the locations of UxNBs and/or allocated power for each user at each UxNB are optimized. We assumed that $K=16$, $M=16$, $S=400$, $N=4$, and $G=9$.}
  \label{R_vs_P_PXY_M16_K16_S400}
\end{figure}

Fig. \ref{R_vs_P_PXY_M16_K16_S400} shows the achievable minimum rate per user in the proposed cell-free scheme versus the total power of each UxNB ($P_1=...=P_M$) for the cases where the locations of UxNBs and/or allocated power for each user at each UxNB are optimized. We set the parameter values as $K=16$, $M=16$, $S=400$, $N=4$, and $G=9$. We consider the same heuristic solutions with Fig. \ref{R_vs_M_PXY_K_16_P25_S400} for power allocation and locations of UxNBs. We can see that optimizing power is more useful for lower transmit powers because there are enough power resources at higher transmit powers for managing fairness among users. Also, as in Fig. \ref{R_vs_M_PXY_K_16_P25_S400}, we can see that when the locations of UxNBs are optimized, equal power allocation is the optimal solution for power.

\begin{figure}
\centering
  \includegraphics[width=\linewidth,height=5cm]{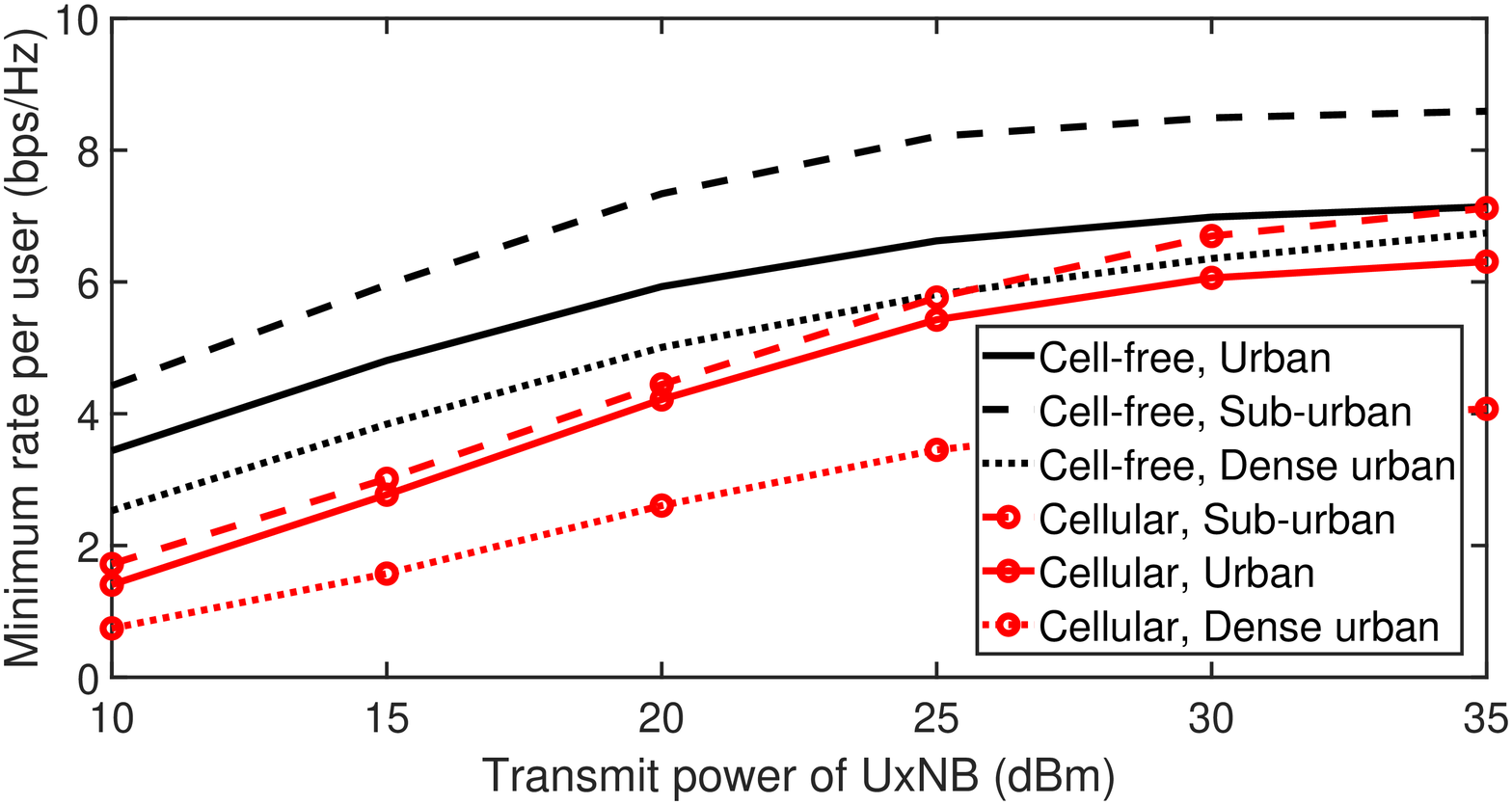}
  \captionof{figure}{The achievable minimum rate per user versus total power of each UxNB ($P_1=...=P_M$) in aerial schemes for three urban, suburban, and dense urban environments. We set $K=16$, $M=16$, $N=4$, $S=400$, and $G=9$.}
  \label{R_vs_P_env_K16_M16_S400}
\end{figure}

Fig. \ref{R_vs_P_env_K16_M16_S400} shows the achievable minimum rate per user versus the total power of each UxNB ($P_1=...=P_M$) in aerial schemes for three urban, suburban, and dense urban environments. According to \cite{Irem}, the excessive path losses affecting the air-to-ground links in LoS and NLoS cases, i.e., $(\eta_{\mathsf{LoS}}^{\mathsf{dB}},\eta_{\mathsf{NLoS}}^{\mathsf{dB}})$, are equal to $(0.1,21)$, $(1,20)$, and $(1.6,23)$ for suburban, urban, and dense urban environments, respectively. Also, parameters  $(A,B)$ are equal to $(4.88,0.43)$, $(9.61,0.16)$, and $(12.8,0.11)$ for suburban, urban, and dense urban environments, respectively \cite{Irem}. 
As we can see, in both cellular and cell-free schemes, the suburban environment outperforms the urban and dense urban environments, and the dense urban environment performs the worst. This is due to the higher probability of establishing a LoS link in a suburban environment. Also, one can see that the proposed cell-free scheme performs better than the baseline cellular scheme for all environments, and this is due to the utilization of all the received signals from all users at each UxNB in the cell-free mode.

\begin{figure}
    \centering
  \includegraphics[width=\linewidth,height=5cm]{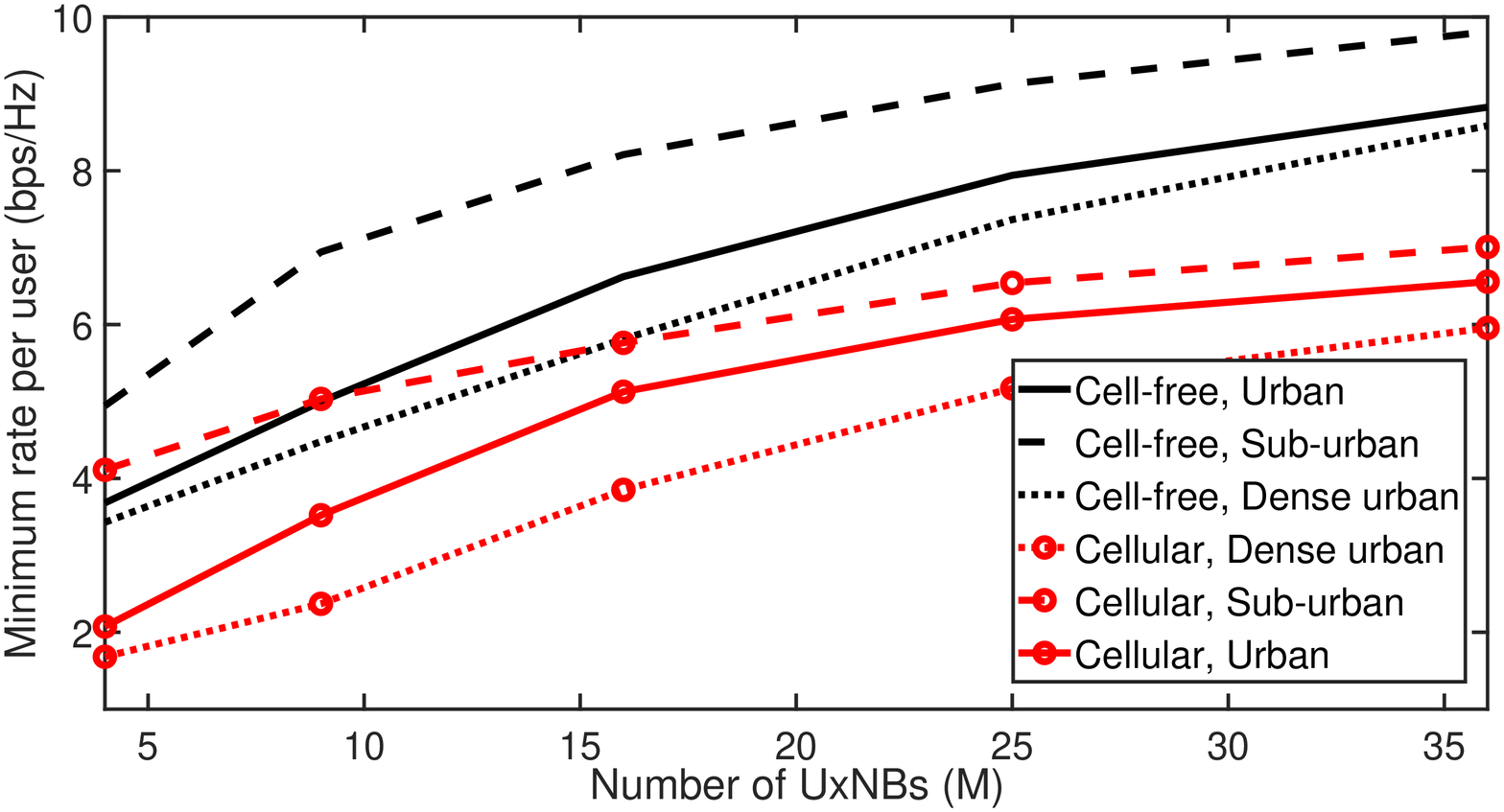}
  \captionof{figure}{The achievable minimum rate per user versus the number of UxNBs ($M$) in aerial schemes for three urban, suburban, and dense urban environments. We set $P_m=25~\mathrm{dBm},~\forall m$, $K=16$, $N=4$, $S=400$, and $G=9$.}
  \label{R_vs_M_env_K16_S400_P25}
\end{figure}

Fig. \ref{R_vs_M_env_K16_S400_P25} shows the achievable minimum rate per user versus the number of UxNBs ($M$) in aerial schemes for three urban, suburban, and dense urban environments.
As we can see, the proposed aerial cell-free scheme outperforms the aerial cellular scheme for all environments. 
We can also see that by increasing $M$, the performance of both cell-free and cellular schemes improves for all environments, and this is due to the establishment of more LoS links between the UxNBs and users for a larger $M$ in aerial schemes. Further, it is shown that the superiority of the proposed cell-free scheme over the aerial cellular scheme for all environments increases by $M$, and this is due to a higher intercell interference for the cellular scheme at a higher value for $M$.


Simulation results showed that the aerial cell-free scheme performs much better than the terrestrial cell-free scheme when the HAPS is equipped with a very large antenna array. 
This performance increase comes at the cost of deploying HAPS and dedicated UAVs. In order to reduce the UAV deployment costs, we can utilize other non-dedicated UAVs as UxNBs by equipping them with the proposed transceiver scheme. In regard to deploying a HAPS, it has to be said that HAPS can be deployed in the stratosphere for many other use cases, such as super macro BS, computing, sensing, and localization, and in this paper, it is utilized as a CPU as well.

\section{Conclusion}
 In this paper, we proposed a cell-free scheme for a set of UxNBs to manage the severe interference in aerial cellular networks between terrestrial users and UxNBs of neighboring cells. We also proposed to use a HAPS as a CPU to combine all the received signals from all UxNBs in the sub-THz band. This involved proposing a transceiver scheme at the UxNBs, a receiver scheme at the HAPS, and formulating an optimization problem to maximize the minimum SINR of users. Simulation results proved the superiority of the proposed scheme compared to aerial cellular and terrestrial cell-free baseline schemes in urban, suburban, and dense urban environments, which is due to the existence of LoS links between users and UxNBs. Simulation results also showed that utilizing a HAPS as a CPU is useful when the considerable path loss in the sub-THz band between UxNBs and the HAPS is compensated for by a high number of antenna elements at the HAPS.

\appendices

\section{PROOF OF Proposition \ref{propos_sinr}}\label{proof_sinr}
In order to derive the $\mathsf{SINR}$ of each user, we rewrite the filtered and combined signal at RB $k$ in the HAPS, i.e., $y^k$ in (\ref{yk}),  as
\footnotesize
\begin{equation}
\begin{split}
   y^k&=\sum_{m=1}^M\sum_{s=1}^S c_{ms}^* y_{s}^k=\sum_{m=1}^M\sum_{s=1}^S c_{ms}^*(G\sum_{m'=1}^M c_{m's}\gamma_{m'}\sqrt{P_{km'}}\\&\times\frac{\sum_{n=1}^N(\sum_{k'=1}^{K}h_{k'm'n}\sqrt{P_{k'}}s_{k'}+z_{m'})\times \frac{h_{km'n}^*}{|h_{km'n}|}}{|\sum_{n=1}^N(\sum_{k'=1}^{K}h_{k'm'n}\sqrt{P_{k'}}s_{k'}+z_{m'})\times \frac{h_{km'n}^*}{|h_{km'n}|}|}+Z_H).
   \end{split}\label{yk_appen}
\end{equation}
\normalsize
Then, with the use-and-then-forget bound \cite{marzetta2016fundamentals}, we derive the achievable SINR. 
From the last equation of (\ref{yk_appen}), we write the desired signal (DS) for user $k$ as
\footnotesize
\begin{equation}\label{ds}
\begin{split}
    \mathsf{DS}_k&=\sum_{m=1}^M\sum_{s=1}^S c_{ms}^*G\sum_{m'=1}^M c_{m's}\gamma_{m'}\sqrt{P_{km'}}\\&\times\frac{\sum_{n=1}^N |h_{km'n}|\sqrt{P_{k}}s_{k}}{|\sum_{n=1}^N(\sum_{k'=1}^{K}h_{k'm'n}\sqrt{P_{k'}}s_{k'}+z_{m'})\times \frac{h_{km'n}^*}{|h_{km'n}|}|}.
\end{split}
\end{equation}
\normalsize
The standard deviation of the power normalization factor, i.e., $P_{\mathrm{NF}}=|\sum_{n=1}^N(\sum_{k'=1}^{K}h_{k'm'n}\sqrt{P_{k'}}s_{k'}+z_{m'})\times \frac{h_{km'n}^*}{|h_{km'n}|}|$, in the denominator of the desired signal in (\ref{ds}) is given by
\footnotesize
\begin{equation} \label{F_norm}
    \begin{split}
&F_{\mathsf{NORM}}=\sqrt{E[P_{\mathrm{NF}} \times P_{\mathrm{NF}}^*]}
       =\frac{\sqrt{N}}{M}\sqrt{\sum_{m=1}^M\sum_{k'=1}^{K}\beta_{k'm}^2P_{k'}+M\sigma^2}.
    \end{split}
\end{equation}
\normalsize
The expectation of the desired signal for user $k$ can then be written as
\footnotesize
\begin{equation}
    \begin{split}
        E[\mathsf{DS}_k]&=E[\frac{1}{F_{\mathsf{NORM}}}\sum_{m=1}^M\sum_{s=1}^S c_{ms}^*G\sum_{m'=1}^M c_{m's}\gamma_{m'}\sqrt{P_{km'}}\sum_{n=1}^N |h_{km'n}|\sqrt{P_{k}}]\\&=\frac{GNS\sqrt{P_{k}}}{F_{\mathsf{NORM}}}\sum_{m=1}^M \gamma_{m}\sqrt{P_{km}} \beta_{km}.
    \end{split}
\end{equation}
\normalsize
Next, we derive the variance of the interference terms in (\ref{yk_appen}). We can see there that we have two types of interference. The first type is caused by interference from users, and we show it with $I_{k,U}$. The second type is due to the amplified noise at UxNBs, and we name it $I_{k,N}$. For $I_{k,U}$, we can write
\footnotesize
\begin{equation}
    \begin{split}\label{var_IU}
        E[I_{k,U}&\times I_{k,U}^*]=E[(\frac{1}{F_{\mathsf{NORM}}}\sum_{m=1}^M\sum_{s=1}^S c_{ms}^*G\sum_{m'=1}^M c_{m's}\gamma_{m'}\sqrt{P_{km'}}\\&\times\sum_{n=1}^N\sum_{k'=1}^{K}h_{k'm'n}\sqrt{P_{k'}}s_{k'}\times \frac{h_{km'n}^*}{|h_{km'n}|})\times(\frac{1}{F_{\mathsf{NORM}}}\sum_{m=1}^M\sum_{s=1}^S c_{ms}^*G\\&\times \sum_{m'=1}^M c_{m's}\gamma_{m'}\sqrt{P_{km'}}\sum_{n=1}^N\sum_{k'=1}^{K}h_{k'm'n}\sqrt{P_{k'}}s_{k'}\times \frac{h_{km'n}^*}{|h_{km'n}|})^*]\\&=\frac{NSG^2}{F_{\mathsf{NORM}}^2}\sum_{m=1}^M\gamma_{m}^2P_{km}\sum_{k'=1}^{K}\beta_{k'm}^2P_{k'}.
    \end{split}
\end{equation}
\normalsize
Also, for the variance of the $I_{k,N}$, we can write
\footnotesize
\begin{equation}
    \begin{split}\label{var_IN}
       & E[I_{k,N}I_{k,N}^*]=\frac{1}{F_{\mathsf{NORM}}^2}E[(\sum_{m=1}^M\sum_{s=1}^S c_{ms}^*G\sum_{m'=1}^M c_{m's}\gamma_{m'}\sqrt{P_{km'}}\sum_{n=1}^Nz_{m'}\\&\times \frac{h_{km'n}^*}{|h_{km'n}|})
        \times(\sum_{m=1}^M\sum_{s=1}^S c_{ms}^*G\sum_{m'=1}^M c_{m's}\gamma_{m'}\sqrt{P_{km'}}\sum_{n=1}^Nz_{m'}\times \frac{h_{km'n}^*}{|h_{km'n}|})^*]\\&=\frac{NSG^2}{F_{\mathsf{NORM}}^2}\sum_{m=1}^M \gamma_{m}^2P_{km}\sigma^2.
    \end{split}
\end{equation}
\normalsize
Finally, we show the noise at the HAPS with $N_{\mathsf{HAPS}}$, whose variance is given by
\footnotesize
\begin{equation}
    \begin{split}
        E[N_{\mathsf{HAPS}}N_{\mathsf{HAPS}}^*]&=E[\sum_{m=1}^M\sum_{s=1}^S c_{ms}^*Z_H\times(\sum_{m=1}^M\sum_{s=1}^S c_{ms}^*Z_H)^*]=MS\sigma_{H}^2.
    \end{split}
\end{equation}
\normalsize
As we mentioned in Section II, in the sub-THz band, the absorbed parts of the signals by the medium are re-emitted with a random phase shift \cite{Petrov_SINR_THz,Saad_thz}. Due to its random phase, this re-emission interference signal, which is indicated by $I_{k,R}$ for user $k$, is uncorrelated with $y^k$ in (\ref{yk_appen}), and its mean value is equal to $0$. In order to get the variance of the re-emission interference for user $k$, we just need to replace the terms $\gamma_m^2= \tau_m\rho_m^2 $ with $(1-\tau_m)\rho_m^2 $ in the variance expressions of the interference formulas in (\ref{yk_appen}), i.e., $E[I_{k,U}\times I_{k,U}^*]$ in (\ref{var_IU}) and  $E[I_{k,N}\times I_{k,N}^*]$ in (\ref{var_IN}), and combine them as follows:
\footnotesize
\begin{equation}
\begin{split}
   E[I_{k,R}I_{k,R}^*]&=\frac{NSG^2}{F_{\mathsf{NORM}}^2}\sum_{m=1}^M(1-\tau_m)\rho_m^2P_{km}\sum_{k'=1}^{K}\beta_{k'm}^2P_{k'}\\&+\frac{NSG^2}{F_{\mathsf{NORM}}^2}\sum_{m=1}^M (1-\tau_m)\rho_m^2P_{km}\sigma^2. 
   \end{split}
\end{equation}
\normalsize

Now, according to the derived formulas for $E[\mathsf{DS}_k]$, $E[I_{k,U}\times I_{k,U}^*]$, $E[I_{k,N}\times I_{k,N}^*]$, $ E[N_{\mathsf{HAPS}}N_{\mathsf{HAPS}}^*]$, and $E[I_{k,R}I_{k,R}^*]$, we derive the $\mathsf{SINR}$ of the user $k$ for the proposed scheme as follows:
\scriptsize
\begin{equation}\label{sinr_k_f}
\begin{split}
    &\mathsf{SINR}_k=\frac{E[\mathsf{DS}_k]^2}{E[I_{k,U} I_{k,U}^*]+E[I_{k,N} I_{k,N}^*]+E[I_{k,R}I_{k,R}^*]+E[N_{\mathsf{HAPS}}N_{\mathsf{HAPS}}^*]}=\\&\frac{G^2N^2SP_{k}(\sum_{m=1}^M \gamma_{m}\sqrt{P_{km}} \beta_{km})^2}{NG^2\sum_{m=1}^M\rho_{m}^2P_{km}\sum_{k'=1}^{K}\beta_{k'm}^2P_{k'}+NG^2\sum_{m=1}^M \rho_{m}^2P_{km}\sigma^2+M\sigma_{H}^2F_{\mathsf{NORM}}^2}.
\end{split}
\end{equation}
\normalsize
By substituting $F_{\mathsf{NORM}}$ from (\ref{F_norm}) in SINR expressions in (\ref{sinr_k_f}), we can derive the achievable rate of user $k$ as $R_k=\log_2(1+\mathsf{SINR}_k)$,
and the proof is completed.

\begin{figure*}
\footnotesize
\begin{equation}\label{uls}
  \begin{split}
      \mathsf{ULS}&(f,t)=\{T:f(T)>t\}=\bigg\{T:\frac{MP^2NSP_{k}(\sum_{m=1}^M \gamma_{m}T_{km} \beta_{km})^2}{MP^2\sum_{m=1}^M\rho_{m}^2T_{km}^2\sum_{k'=1}^{K}\beta_{k'm}^2P_{k'}+MP^2\sum_{m=1}^M \rho_{m}^2T_{km}^2\sigma^2+\sigma_{H}^2(\sum_{m=1}^M\sum_{k'=1}^{K}\beta_{k'm}^2P_{k'}+M\sigma^2)}>t,~\forall k\bigg\}\\&=\bigg\{T:\sqrt{MP^2\sum_{m=1}^M\rho_{m}^2T_{km}^2\sum_{k'=1}^{K}\beta_{k'm}^2P_{k'}+MP^2\sum_{m=1}^M \rho_{m}^2T_{km}^2\sigma^2+\sigma_{H}^2(\sum_{m=1}^M\sum_{k'=1}^{K}\beta_{k'm}^2P_{k'}+M\sigma^2)}<\frac{\sqrt{MP^2NSP_{k}}\sum_{m=1}^M \gamma_{m}T_{km} \beta_{km}}{\sqrt{t}},~\forall k\bigg\}.
  \end{split}  
\end{equation}
\normalsize
\end{figure*}

\section{PROOF OF Proposition \ref{propos_quasi}}\label{proof_quasi}
In order to prove the quasi-concavity of (P1), we need to show that the objective function is quasi-concave, and that the constraints' set is convex. To prove the quasi-concavity of the objective function, we just need to prove that its upper-level set is a convex set \cite{boyd2020disciplined}. For this, we first perform the variable change $\bold{T}=[P_{km}^2]_{K\times M}$ in (P1) and show the objective function with $f(T)$. Hence, for any $t\in \mathbb{R}_+$, the upper-level set (ULS) of the objective function is given by (\ref{uls}) which is in the form of a norm function less than an affine function of variable $T$, and hence it is a convex set. With the new variable $T$, the constraint (\ref{eq:constraint-sum_P}) will be $\sum_{k=1}^KT_{km}^2\leq P_m, ~\forall m$, which is a convex set, and the proof is completed.




\ifCLASSOPTIONcaptionsoff
  \newpage
\fi



\bibliographystyle{IEEEtran}
\bibliography{IEEEabrv,myref}
%

\end{document}